\documentclass[11pt]{article}

\usepackage[margin=1in]{geometry}
\usepackage{amsmath,amssymb,amsthm,mathtools,bm}
\usepackage{microtype}
\usepackage{array}
\usepackage{enumitem}
\usepackage{xcolor}
\usepackage{tikz}
\usepackage{aliascnt}
\usepackage[hidelinks]{hyperref}

\newtheorem{theorem}{Theorem}[section]
\newaliascnt{lemma}{theorem}
\newtheorem{lemma}[lemma]{Lemma}
\aliascntresetthe{lemma}
\newaliascnt{proposition}{theorem}
\newtheorem{proposition}[proposition]{Proposition}
\aliascntresetthe{proposition}
\newaliascnt{corollary}{theorem}
\newtheorem{corollary}[corollary]{Corollary}
\aliascntresetthe{corollary}
\theoremstyle{definition}
\newaliascnt{definition}{theorem}

\aliascntresetthe{definition}
\newaliascnt{example}{theorem}

\aliascntresetthe{example}
\theoremstyle{remark}
\newaliascnt{remark}{theorem}

\aliascntresetthe{remark}

\usepackage[nameinlink,capitalize,noabbrev]{cleveref}

\newcommand{\Prob}{\mathbb P}
\newcommand{\one}{\mathbf 1}

\newcommand{\supp}{\operatorname{supp}}

\newcommand{\TT}{\mathrm{TT}}

\newcommand{\cE}{\mathcal E}

\title{Beyond Cut Balance:\\
Spectral Sparsification of the Nonlinear Directed Laplacian}
\author{Yuichi Yoshida\thanks{Supported by JSPS KAKENHI Grant Number 26K21777, 26K21940, 25K24540.}\\
National Institute of Informatics\\
\texttt{yyoshida@nii.ac.jp}
}

\begin{document}
\maketitle

\begin{abstract}
Digraphs with constant cut balance admit nearly linear directed cut
sparsifiers. This condition requires the total arc weights in the two
directions of every cut to be within a constant factor of each other.
We ask whether this condition also permits nearly linear spectral
sparsification with respect to the energy of the nonlinear directed Laplacian.
For a weighted digraph $G=(V,E,w)$, let
\[
    Q_G^+(x)=\sum_{(u,v)\in E}w_{uv}(x_u-x_v)_+^2,
    \qquad (t)_+:=\max\{t,0\}.
\]
This energy agrees with the outgoing-cut function on binary vectors.  A
spectral sparsifier is a nonnegatively reweighted subgraph that preserves
$Q_G^+(x)$ within a factor of $1\pm\varepsilon$ simultaneously for all
$x\in\mathbb R^V$.

We show that cut balance alone does not yield nearly linear spectral
sparsifiers: for constant error, the worst-case support size for simple
unweighted Eulerian digraphs is $\widetilde\Theta(n^{3/2})$, although Eulerian
digraphs are perfectly cut-balanced and admit nearly linear directed cut
sparsifiers.  In contrast, we prove that every $n$-vertex tournament has a
spectral sparsifier with $\widetilde O(n/\varepsilon^3)$ arcs, without any
assumption on its cut balance.  This includes the transitive tournament, whose cut balance
is unbounded.  Thus perfect balance does not guarantee nearly linear spectral
sparsification, while unbounded imbalance does not preclude it.

Finally, we use convex duality to show that preserving $Q_G^+$ also preserves,
for every feasible demand vector, the optimum quadratic cost of a nonnegative
flow.  Hence the guarantee contains information beyond directed cut values.
\end{abstract}
\thispagestyle{empty}
\clearpage
\tableofcontents
\thispagestyle{empty}
\clearpage
\setcounter{page}{1}
\section{Introduction}

General digraphs may require a quadratic number of arcs for directed cut
sparsification~\cite{CohenEtAl2017}, but cut balance identifies a broad class
for which nearly linear sparsification is possible.  For a weighted digraph
$G$ and a vertex set $S$, let $\delta_G^+(S)$ and $\delta_G^-(S)$ denote the arcs leaving and
entering $S$, respectively.  The \emph{cut balance} of
$G$~\cite{EneEtAl2016} is
\[
    \beta(G):=\inf\left\{\beta\ge1:
       w(\delta_G^+(S))\le \beta w(\delta_G^-(S))
       \text{ for every }S\subseteq V(G)\right\},
\]
with $\beta(G)=\infty$ if no finite value satisfies these inequalities.
Applying the inequality also to $V(G)\setminus S$ compares the two directions
of every cut.  Building on the cut sparsifiers of Ikeda and
Tanigawa~\cite{IkedaTanigawa2018}, Cen, Cheng, Panigrahi, and Sun proved that
a digraph with
$\beta(G)\le\beta$ has a weighted subgraph with
$\widetilde O(\beta n/\varepsilon^2)$ arcs that preserves all outgoing cuts
within a factor of $1\pm\varepsilon$~\cite{CenEtAl2021}.  In particular, an
Eulerian digraph has $\beta(G)=1$ and therefore admits a nearly linear cut
sparsifier.

Cen et al. asked whether directed cut sparsification extends to a spectral
notion that still preserves directed cuts.  They proposed preserving
\cite[Section~7]{CenEtAl2021}
\begin{equation}
    \label{eq:intro-directed-energy}
    Q_G^+(x)
    :=
    \sum_{(u,v)\in E(G)}w_{uv}(x_u-x_v)_+^2,
    \qquad (t)_+:=\max\{t,0\},
\end{equation}
for every real potential $x\in\mathbb R^{V(G)}$.  An arc $u\to v$ contributes
only when the potential drops in its direction.  For a vertex set $S$,
\[
    Q_G^+(\mathbf 1_S)=w(\delta_G^+(S)),
\]
so \eqref{eq:intro-directed-energy} extends the outgoing-cut function from
binary indicators to arbitrary real potentials.  The same energy is
associated with the nonlinear directed Laplacian introduced by
Yoshida~\cite{Yoshida2016}.

For comparison, let $L_G$ be the Laplacian matrix of a weighted
undirected graph $G$.  Its energy is the quadratic form
\[
    x^\top L_Gx=\sum_{\{u,v\}\in E(G)}w_{uv}(x_u-x_v)^2,
\]
and every $n$-vertex weighted graph has a
$(1\pm\varepsilon)$-spectral sparsifier with
$O(n/\varepsilon^2)$ edges~\cite{BatsonSpielmanSrivastava2012}.
The squared difference $(x_u-x_v)^2$ is invariant under exchanging $u$ and
$v$, so it cannot distinguish the two orientations of an edge.  In
\eqref{eq:intro-directed-energy}, the term $(x_u-x_v)_+^2$ instead contributes
only when the potential decreases along the arc $u\to v$.

There is also a distinct matrix-based theory of directed Laplacian
sparsification~\cite{CohenEtAl2017,ChuEtAl2018,SachdevaThudiZhao2024,JambulapatiEtAl2025}.
For Eulerian digraphs, it gives nearly linear-size approximations of a
nonsymmetric Laplacian matrix relative to its symmetrization.  These
approximations in particular preserve directed cuts, but they do not require
pointwise preservation of $Q_G^+$.  Throughout this paper, a
\emph{$(1\pm\varepsilon)$-spectral sparsifier} of $G$ is a nonnegatively
reweighted subgraph $H\subseteq G$ such that
\begin{equation}
    \label{eq:intro-sparsifier-guarantee}
    (1-\varepsilon)Q_G^+(x)
    \le Q_H^+(x)
    \le(1+\varepsilon)Q_G^+(x)
    \qquad\forall x\in\mathbb R^{V(G)}.
\end{equation}
Unlike in the undirected case, the set of contributing arcs changes with the
ordering of the vertex potentials: an arc $u\to v$ contributes only when
$x_u>x_v$.  A sparsifier must therefore approximate the energy for all these
orderings simultaneously.  In particular, preserving $Q_G^+$ imposes
constraints from multivalued potentials beyond those detected by cut
indicators.

No subquadratic guarantee is possible for general digraphs, even if one tests
only cut indicators.  Let $A$ and $B$ be two sets of $k$ vertices and include
every arc from $A$ to $B$.  For each $a\in A$ and $b\in B$, the outgoing cut
of $\{a\}\cup(B\setminus\{b\})$ consists only of $a\to b$.  Every
multiplicative sparsifier with error smaller than one must therefore retain
all $k^2=\Theta(n^2)$ arcs~\cite{CohenEtAl2017}.  This example makes it
necessary to identify graph classes with additional structure.

We show that cut balance alone cannot yield a nearly linear answer to this
question and that cut balance is not necessary for nearly linear
sparsification.  We prove these two statements through simple Eulerian
digraphs and tournaments, respectively.

\subsection{Our results}

This subsection states the worst-case sparsifier size for simple Eulerian
digraphs, the upper bound for tournaments, and a quadratic-flow consequence
of preserving $Q_G^+$.

Unless stated otherwise, the digraphs in our graph-class theorems are
unweighted; their sparsifiers may be weighted.
A digraph is simple if it has no loops and at most one arc for each ordered
pair of distinct vertices; both $u\to v$ and $v\to u$ are allowed.
Throughout the paper, $\widetilde O$, $\widetilde\Omega$, and
$\widetilde\Theta$ suppress factors polylogarithmic in the positive parameters
displayed in their arguments.  A subscript lists parameters treated as fixed.

For a directed graph $G$, let $\operatorname{spn}_\varepsilon(G)$ denote the
minimum number of arcs in a $(1\pm\varepsilon)$-spectral sparsifier of $G$.
The formal definition appears in \cref{sec:general}.

\paragraph{Eulerian digraphs: a separation between cuts and potentials.}
An Eulerian digraph has equal total arc weight in the two directions of every
cut.  Nevertheless, we show that this exact balance does not give a nearly
linear worst-case bound for preserving all real potentials.

\begin{theorem}[Worst-case sparsifier size for simple Eulerian digraphs]
\label{thm:eulerian-n32-transition}
Fix $0<\varepsilon_0\le1/2$.  For all sufficiently large $n$ and all integers
$3\le m\le n(n-1)$, every simple Eulerian digraph on $n$ vertices with at
most $m$ arcs has a $(1\pm\varepsilon_0)$-spectral sparsifier with
$\widetilde O_{\varepsilon_0}(\min\{m,n^{3/2}\})$ arcs.  Conversely, there
exists a simple Eulerian digraph on $n$ vertices with at most $m$ arcs such
that every $(1\pm\varepsilon_0)$-spectral sparsifier has
$\widetilde\Omega_{\varepsilon_0}(\min\{m,n^{3/2}\})$ arcs.  Isolated
vertices are allowed.
\end{theorem}

The lower bound concerns only the number of retained arcs;
the sparsifier itself need not be Eulerian.  Combining our lower-bound
construction with the cut-sparsification theorem of Cen et
al.~\cite[Theorem~3]{CenEtAl2021} gives the following separation on the same
family of input digraphs.

\begin{corollary}[Cut--spectral separation for Eulerian digraphs]
\label{cor:eulerian-cut-spectral-separation}
For every fixed $0<\varepsilon\le1/2$, there are infinitely many $n$ and
simple unweighted strongly connected Eulerian digraphs $G$ on $n$ vertices
with the following properties.  The digraph $G$ has a nonnegatively
reweighted $(1\pm\varepsilon)$-cut sparsifier with
$\widetilde O_\varepsilon(n)$ arcs, whereas every
$(1\pm\varepsilon)$-spectral sparsifier of $G$ has
$\Omega(n^{3/2})$ arcs.
\end{corollary}

\paragraph{Tournaments: nearly linear upper bounds without cut balance.}
A tournament contains exactly one orientation of every pair of distinct
vertices.  Our upper bound requires no assumption of regularity, strong
connectivity, or cut balance.

\begin{theorem}[Sparsification of tournaments]
\label{thm:tournament-sparsification}
For every tournament $T$ on $n$ vertices and every $0<\varepsilon<1/2$,
\[
    \operatorname{spn}_\varepsilon(T)
    \le
    \widetilde O\left(
       \min\left\{n^2,\frac{n}{\varepsilon^3}\right\}
    \right).
\]
\end{theorem}

The transitive tournament $\TT_n$ has vertex set $[n]$ and the arc $i\to j$
whenever $i<j$.  For every nonempty proper prefix of $[n]$, all crossing arcs
leave the prefix and none enter it, so $\beta(\TT_n)=\infty$.  Thus
\cref{thm:tournament-sparsification} applies even when cut balance is
unbounded.  Together with the Eulerian lower bound in
\cref{thm:eulerian-n32-transition}, this shows that perfect
cut balance does not guarantee a nearly linear sparsifier for $Q_G^+$, while
unbounded cut balance does not preclude one.  For constant error, the
Eulerian lower bound is $\widetilde\Omega(n^{3/2})$, and the complete-bipartite
DAG requires $\Omega(n^2)$ arcs, whereas every tournament admits a
$\widetilde O(n)$-arc sparsifier.

\paragraph{Quadratic costs of nonnegative flows.}
As a consequence of spectral approximation, our sparsifiers also preserve
the minimum quadratic cost of routing every feasible demand with nonnegative
arc flows.  Let
$B_G\in\mathbb R^{V\times E}$ be the vertex--arc incidence matrix whose column
for $e=(u,v)$ is $\mathbf 1_u-\mathbf 1_v$.  For $b\in\mathbb R^V$, define
\begin{equation}
    \label{eq:intro-flow-cost}
    \mathcal F_G(b)
    :=
    \min_{\substack{f\in\mathbb R_{\ge0}^E\\B_Gf=b}}
       \sum_{e\in E}\frac{f_e^2}{w_e},
\end{equation}
and set $\mathcal F_G(b)=+\infty$ if no nonnegative flow routes $b$.
This flow problem is dual to an optimization problem involving $Q_G^+$,
a relation studied by Fujii, Soma, and Yoshida~\cite{FujiiSomaYoshida2021}.
Applying this duality gives the following consequence of spectral approximation.

\begin{corollary}[Preservation of quadratic nonnegative-flow costs]
\label{cor:flow-cost-preservation}
Let $0<\varepsilon<1$, let $G$ be a weighted digraph with positive arc
weights, and let $H$ be a $(1\pm\varepsilon)$-spectral sparsifier of $G$.
If $b$ is feasible in $G$, then it is feasible in $H$ and
\[
    \frac{\mathcal F_G(b)}{1+\varepsilon}
    \le \mathcal F_H(b)
    \le \frac{\mathcal F_G(b)}{1-\varepsilon}.
\]
\end{corollary}

The corollary concerns optimum objective values for all demands; it does not
assert that the corresponding optimum flows are close or that the sparsifier
can be constructed within the time needed to solve a flow problem.  Its proof
in \cref{subsec:flow-costs} follows from convex duality.

\subsection{Technical overview}

This subsection explains the Eulerian lower bound using multivalued
potentials and the matching upper bound using short directed cycles and row
sampling.  We then show why this sampling theorem gives no bound even for
the transitive tournament and explain how directed paths constrain arc
energies enough to obtain nearly linear sparsifiers for all tournaments.

\subsubsection{Eulerian digraphs: matching lower and upper bounds}

We explain why the worst-case sparsifier size for simple Eulerian digraphs
is $\widetilde\Theta(n^{3/2})$ for constant error.

The complete bipartite DAG is hard to sparsify because every arc can be
isolated by a potential.  Indeed, for an arc $e=a\to b$, the cut potential
$\mathbf 1_{\{a\}\cup(B\setminus\{b\})}$ gives positive energy to $e$ and zero
energy to every other arc.  If a multiplicative sparsifier omits $e$, its
energy on this potential is zero whereas the energy of the original graph is
positive.  This argument applies separately to all $\Theta(n^2)$ arcs.  Thus
the obstruction is not merely that one arc can contribute a large fraction
of the energy; it is that quadratically many arcs can be isolated one at a
time.

\paragraph{Why cut indicators do not detect the Eulerian lower bound.}
The complete bipartite DAG is far from Eulerian: all its arcs leave one side
and enter the other.  Adding return arcs can make an isolating potential
contribute energy outside the original bipartite subgraph.  The lower-bound
construction spreads this additional energy over many small potential drops.

Partition the vertices into $k\ge4$ layers $V_0,\ldots,V_{k-1}$ of size $b$
and include every arc from $V_i$ to $V_{i+1\bmod k}$.  Every vertex has $b$
incoming and $b$ outgoing arcs, so the resulting digraph
$G^{\mathrm{cyc}}_{k,b}$ is Eulerian.  It has $n=kb$ vertices and $m=kb^2$
arcs.

Fix a vertex $u\in V_0$ and a subset $T\subseteq V_1$.  Define the potential
$x^{u,T}$ by setting $x^{u,T}_u=1$, setting $x^{u,T}_v=0$ for
$v\in(V_0\setminus\{u\})\cup T$, and setting $x^{u,T}_v=1$ for
$v\in V_1\setminus T$, and setting
$x^{u,T}_v=(k-r)/(k-2)$ for $v\in V_r$ and $2\le r\le k-1$.  Thus precisely
the arcs from $u$ to $T$ have a unit
drop between $V_0$ and $V_1$.  All vertices in $V_2$ have value one, and the
values decrease linearly through the remaining layers toward value zero on
$V_0\setminus\{u\}$.  The unit decrease needed to return around the cycle of
layers is therefore divided into $k-2$ equal parts.  Because the energy
squares each decrease, the energy outside the arcs from $u$ to $T$ is
independent of $T$ and equals
\[
    \Gamma_{k,b}
    :=\frac{b^2}{k-2}-\frac{b}{(k-2)^2}
    =\Theta(b^2/k).
\]
Consequently,
\begin{equation}
    \label{eq:overview-witness-energy}
    Q_{G^{\mathrm{cyc}}_{k,b}}^+(x^{u,T})=\Gamma_{k,b}+|T|.
\end{equation}

Suppose a sparsifier omits the arcs from $u$ to a set $Z\subseteq V_1$.
It assigns the same energy to $x^{u,\varnothing}$ and $x^{u,Z}$ because only
the omitted arcs distinguish them.  By
\eqref{eq:overview-witness-energy}, the original graph assigns energies
$\Gamma_{k,b}$ and $\Gamma_{k,b}+|Z|$.  The approximation inequalities in
\eqref{eq:intro-sparsifier-guarantee}, applied to these two potentials, force
\begin{equation}
    \label{eq:overview-omission-bound}
    |Z|\le \frac{2\varepsilon}{1-\varepsilon}\Gamma_{k,b}.
\end{equation}
By cyclic symmetry, the same argument applies to every vertex in every layer.
For $0<\varepsilon\le1/2$ and $k=8b+2$, we have
$\Gamma_{k,b}<b/8$, so every vertex retains at least $3b/4$ outgoing arcs.
Since
\[
    n=kb=\Theta(b^2),
    \qquad
    m=kb^2=\Theta(b^3)=\Theta(n^{3/2}),
\]
every constant-error sparsifier retains $\Omega(n^{3/2})$ arcs.  For
$m=O(n)$, a directed cycle padded with isolated vertices gives the
$\Omega(m)$ lower bound.  For larger values of $m$, adjusting $k$ and $b$ and
padding with isolated vertices gives the $\Omega(\min\{m,n^{3/2}\})$ lower
bound.

Eulerian digraphs admit nearly linear directed cut
sparsifiers~\cite{CenEtAl2021}.  Cut indicators take only the values zero and
one, so they cannot realize the gradual return decrease used in our lower
bound.  A spectral sparsifier must also preserve the energies of multivalued
potentials.  These can keep unit drops on selected arcs leaving one vertex
while spreading the return decrease over many small steps, reducing the
energy contributed by the remaining arcs.  A subgraph that omits many
selected arcs must then incur a large relative error on some potential.
Applying this argument at every vertex of our construction forces a constant
fraction of all arcs to remain.

\paragraph{Eulerian upper bound by removing short directed cycles.}
For the matching upper bound, we apply row sampling to arcs decomposed into
short directed cycles.  For a matrix $A$ with rows $a_i^\top$, a row-sampling
theorem of Munteanu and Omlor~\cite{MunteanuOmlor2024} preserves
$\sum_i(a_i^\top x)_+^2$ using
\begin{equation}
    \label{eq:overview-row-sampling}
    \widetilde O\left(\frac{\operatorname{rank}(A)\mu(A)}{\varepsilon^2}\right)
    \quad\text{rows},
    \qquad
    \mu(A):=
    \sup_{Ax\ne0}\frac{\|Ax\|_2^2}{\|(Ax)_+\|_2^2},
\end{equation}
provided that $\mu(A)$ is finite.  Let $A_G$ be the arc--vertex incidence
matrix whose row for $u\to v$ is
$\mathbf 1_u^\top-\mathbf 1_v^\top$; then
$Q_G^+(x)=\|(A_Gx)_+\|_2^2$.

Short directed cycles bound the ratio between the symmetric and directed
energies.  Consider a directed cycle
$v_1\to\cdots\to v_\ell\to v_1$ and write
$d_i=x_{v_i}-x_{v_{i+1}}$, with $v_{\ell+1}=v_1$.  Since
$\sum_i d_i=0$, the sum of the positive differences equals the sum of the
absolute values of the negative differences.  Cauchy--Schwarz then gives
\begin{equation}
    \label{eq:overview-cycle-energy}
    \sum_{i=1}^{\ell}d_i^2
    \le \ell\sum_{i=1}^{\ell}(d_i)_+^2.
\end{equation}
If the arcs of a digraph $F$ are partitioned into directed cycles of length
at most $L$, summing \eqref{eq:overview-cycle-energy} over those cycles
gives $\|A_Fx\|_2^2\le LQ_F^+(x)$, and hence $\mu(A_F)\le L$.

An arbitrary cycle decomposition of an Eulerian digraph may contain cycles of
length $\Theta(n)$, which would give a quadratic bound.  We obtain shorter
cycles by peeling them while the graph is dense.  Starting from a simple
Eulerian digraph, repeatedly remove a directed cycle until the residual
$R_{\mathrm{res}}$ has at most $M$ arcs, and let $F$ contain the removed arcs.
Removing a cycle
preserves the Eulerian property.  When the current residual has $m_t>M$ arcs,
the theorem of Huang, Ma, Shapira, Sudakov, and
Yuster~\cite{HuangMaShapiraSudakovYuster2013} gives a directed cycle of length
at most $6n^2/m_t$.  Thus the removed arcs are partitioned into cycles of
length
\[
    L=O(1+n^2/M).
\]

Substituting $\mu(A_F)\le L=O(1+n^2/M)$ and
$\operatorname{rank}(A_F)\le n-1$ into \eqref{eq:overview-row-sampling}
gives a sparsifier of $F$ with
\[
    \widetilde O\left(
       \frac{n}{\varepsilon^2}
       +\frac{n^3}{M\varepsilon^2}
    \right)
\]
arcs.  We retain the residual $R_{\mathrm{res}}$ without changing its weights.
The total
support is therefore
\begin{equation}
    \label{eq:overview-eulerian-tradeoff}
    \widetilde O\left(
       M+\frac{n}{\varepsilon^2}
        +\frac{n^3}{M\varepsilon^2}
    \right).
\end{equation}
The terms $M$ and $n^3/(M\varepsilon^2)$ in
\eqref{eq:overview-eulerian-tradeoff} represent the two sides of the
tradeoff: peeling longer leaves fewer residual arcs but permits longer cycles
in $F$.  Choosing $M\asymp n^{3/2}/\varepsilon$ balances these terms.  If the
input has fewer arcs, we retain all of them.  For constant error this gives
the upper bound $\widetilde O(\min\{m,n^{3/2}\})$
(\cref{thm:eulerian-n32-transition}).  In particular, the exponent $3/2$ is the
balance point between the residual size and the cost of sparsifying the
removed cycles.

\subsubsection{Tournaments: nearly linear sparsification}

We now turn to the nearly linear upper bound for tournaments.  The
row-sampling bound in \eqref{eq:overview-row-sampling} need not give a useful
sparsifier for tournaments.  For the transitive tournament $\TT_n$, the
potential $x_i=i$ satisfies $Q_{\TT_n}^+(x)=0$ while
$\|A_{\TT_n}x\|_2^2>0$.  Hence $\mu(A_{\TT_n})=\infty$, and
\eqref{eq:overview-row-sampling} gives no sparsification.  The tournament
upper bound must instead use how its directed paths constrain the arc
energies jointly.

\paragraph{A sampling criterion for all potentials.}
For an arc $e=(u,v)$, write
\[
    q_e^+(x):=w_e(x_u-x_v)_+^2.
\]
Every potential with $Q_G^+(x)>0$ determines a normalized arc-energy vector
\[
    r_x(e):=\frac{q_e^+(x)}{Q_G^+(x)},
    \qquad
    \mathcal R_G:=\{r_x:Q_G^+(x)>0\}\subseteq\Delta_E,
\]
where $\Delta_E$ is the set of nonnegative vectors indexed by the arcs whose
coordinates sum to one.
Thus $r_x(e)$ is the fraction of the total energy contributed by $e$.  A
vector $r_x$ supported on one arc records exactly the isolation phenomenon in
the complete bipartite example.

Retain each arc $e$ independently with probability $\pi_e\in(0,1]$, and let
$\xi_e$ indicate whether it is retained.  If every retained arc $e$ is
assigned weight $w_e/\pi_e$, then the worst relative error is
\begin{equation}
    \label{eq:overview-sampling-error}
    Z(\xi):=
    \sup_{r\in\mathcal R_G}
    \left|\sum_e\left(\frac{\xi_e}{\pi_e}-1\right)r_e\right|.
\end{equation}
For each fixed $x$, the sampled energy is an unbiased estimate of
$Q_G^+(x)$.  We need $Z(\xi)\le\varepsilon$ to obtain the guarantee in
\eqref{eq:intro-sparsifier-guarantee} simultaneously for all potentials.
The probabilities $\pi_e$ determine the expected number
$\sum_e\pi_e$ of retained arcs, whereas the set $\mathcal R_G$ determines
whether that sampling budget suffices.  The quantity
$\sup_{r\in\mathcal R_G}r_e$ is the largest fraction of the energy that arc
$e$ can contribute, but these coordinatewise maxima do not say which arcs can
contribute large fractions for the same potential.  Our tournament argument
uses this dependence among arc contributions.

To bound $Z(\xi)$ in \eqref{eq:overview-sampling-error}, compare two
independent samples $\xi$ and $\xi'$ and
put $\zeta_e=\mathbf 1\{\xi_e\ne\xi'_e\}$.  Conditional on $\zeta$, each
nonzero difference $\xi_e-\xi'_e$ is independently $+1$ or $-1$, with equal
probability; these are called \emph{Rademacher signs}.  A standard comparison
bounds the expected supremum with these signs by a constant times the
corresponding expectation with independent
standard Gaussian coefficients.  This is the \emph{Rademacher--Gaussian
comparison}~\cite[Eq.~(4.8)]{LedouxTalagrand1991}.  Thus the expected error
$\mathbb E Z(\xi)$ is at most a universal constant times
$\mathfrak G_G(\pi)$, where
\begin{equation}
    \label{eq:overview-gaussian-width}
    \mathfrak G_G(\pi)
    :=
    \mathbb E_{\zeta,g}
    \sup_{r\in\mathcal R_G\cup\{0\}}
    \left|\sum_e g_e\frac{\zeta_e}{\pi_e}r_e\right|,
\end{equation}
where the $g_e$ are independent standard Gaussian variables, independent of
$\zeta$.  For fixed $\zeta$, the expectation over $g$ in
\eqref{eq:overview-gaussian-width} measures how large this random linear
function can be when we choose $r\in\mathcal R_G$ after seeing the
coefficients.  This expected supremum is called the \emph{Gaussian width}.
The formal definition of
$\mathfrak G_G(\pi)$ appears in \eqref{eq:restricted-gaussian-width}.
The order of expectations allows the estimate for each fixed $\zeta$ to use
only the coordinates that remain after the Bernoulli restriction.

We show that
$\mathfrak G_G(\pi)\le c\varepsilon$, for a sufficiently small universal
constant $c>0$, suffices for the existence of a
$(1\pm\varepsilon)$-sparsifier with $O(\sum_e\pi_e)$ arcs
(\cref{thm:qgld-general}).  The tournament
upper bound therefore reduces to choosing probabilities with small total sum
and bounding the expected supremum in \eqref{eq:overview-gaussian-width}.

\paragraph{Sampling a transitive tournament.}
We first consider $\TT_n$, whose arcs are $i\to j$ for $i<j$.
Its many short directed paths suggest a simple sampling rule.  For every
intermediate vertex $i<k<j$, every potential satisfies
\begin{equation}
    \label{eq:overview-two-arc}
    (x_i-x_k)_+^2+(x_k-x_j)_+^2
    \ge \frac12(x_i-x_j)_+^2.
\end{equation}
These two-arc paths are edge-disjoint.  Summing their inequalities and
including the contribution of $i\to j$ itself gives
\begin{equation}
    (x_i-x_j)_+^2\le\frac{2}{j-i+1}Q_{\TT_n}^+(x).
    \label{eq:overview-arc-energy-bound}
\end{equation}
Thus an arc of endpoint distance $\Theta(D)$ contributes at most an
$O(1/D)$ fraction of the total energy.  The complete bipartite DAG has
no corresponding two-arc paths, so an individual arc can be isolated there.
Group the arcs by endpoint distance, using $D=1,2,4,\ldots$ and
\[
    E_D:=\{(i,j):D\le j-i<2D\}.
\]
The arc-energy bound \eqref{eq:overview-arc-energy-bound} suggests the
inclusion probability
$p_D=\min\{1,\rho/D\}$, where $\rho\ge1$ controls the sampling budget.
Since $|E_D|\le nD$, the total expected sample size is $O(n\rho\log(2n))$.

The difficulty is to justify this rule simultaneously for all potentials.
The bound \eqref{eq:overview-arc-energy-bound} controls each arc's largest
possible energy fraction, but does not say which arcs can contribute large
fractions together.
Moreover, the potential maximizing the error can be chosen after seeing
the entire sample, and its value order need not agree with the vertex order.
Our proof uses the energy constraints imposed by the short paths and the
vertex order to limit this freedom to adapt to the sampling error.

The sampling criterion in \cref{thm:qgld-general} reduces this task to
bounding an expected Gaussian supremum as in
\eqref{eq:overview-gaussian-width}.  For one range $E_D$,
set $\pi_e=p_D$ on $E_D$ and $\pi_e=1$ elsewhere.  Ranges with $p_D=1$
have zero error, so consider $p_D<1$.  The coefficient of an arc in the
Gaussian sum is $\zeta_e g_e/p_D$, where $\zeta_e$ records disagreement
between two independent Bernoulli$(p_D)$ samples, and the $g_e$ are
independent standard Gaussian variables, independent of the samples.
We first fix both $\zeta$ and $g$ and bound the sum divided by
$Q_{\TT_n}^+(x)$ uniformly over all potentials of positive energy.

We use two decompositions, based on different orders.  Vertex blocks are
consecutive in the fixed order $1<\cdots<n$ and are chosen independently
of $x$.  The value tree instead uses the sorted distinct values attained
by $x$.  In this overview, we keep all vector coordinates in the original
vertex order.

First partition the vertices into blocks of size $\Theta(D)$, with a
small enough constant that every arc in $E_D$ has $\Theta(D)$ vertices
strictly between its endpoint blocks.  For each endpoint-block pair,
write $I$ for the earlier block, $J$ for the later block, and $W$ for the
vertex interval containing both blocks and all intervening vertices
(\cref{sec:topological-windows}).

For a fixed potential, list its distinct attained values as
$y_0<\cdots<y_m$ and build a balanced binary tree whose leaves, in
increasing value order, are the gaps $[y_{r-1},y_r]$.  Each node represents
the value interval $K\subseteq\mathbb R$ spanned by its descendant leaves.
The tree organizes the cross terms in squared differences through the identity
\begin{equation}
    \label{eq:overview-square-split}
    (a-b)^2=(a-c)^2+(c-b)^2+2(a-c)(c-b)
    \qquad(b\le c\le a).
\end{equation}
This yields an exact sum of leaf-gap squares and products across split
values (\cref{lem:value-tree-square}).  Let $x^K$ be obtained by clipping
each value $x_i$ to the nearest point in $K$, keeping it at the same
vertex $i$.  At any tree depth the value intervals have disjoint interiors,
giving
\begin{equation}
    \label{eq:overview-clipped-energy}
    \sum_{K\text{ at a fixed depth}}Q_{\TT_n}^+(x^K)
    \le Q_{\TT_n}^+(x).
\end{equation}
Thus the tree has only $O(\log n)$ levels of energy to account for,
irrespective of the spacing of the attained values or their order along
the vertices.

Fix a value node $K$ and an endpoint-block pair $I,J$.  If $c$ is the
node's split value (the midpoint for a leaf gap), put
\[
    \alpha_i=(x_i^K-c)_+\quad(i\in I),
    \qquad
    \beta_j=(c-x_j^K)_+\quad(j\in J).
\]
The contribution of this piece is a constant multiple of the bilinear form
\[
    \alpha^\top X\beta
    =\sum_{i\in I,\,j\in J}X_{ij}\alpha_i\beta_j,
    \qquad
    X_{ij}=\begin{cases}
        \zeta_{ij}g_{ij}/p_D,&(i,j)\in E_D,\\
        0,&\text{otherwise}.
    \end{cases}
\]
Here $X$ contains the random coefficients, whereas $\alpha$ and $\beta$
record the endpoint deviations above and below $c$.  Our local target is
\[
    |\alpha^\top X\beta|
    \le \mathcal A\,E_{K,W}(x),
    \qquad E_{K,W}(x):=Q_{\TT_n[W]}^+(x^K),
\]
with one random factor $\mathcal A$ valid simultaneously for every
potential, value node, and endpoint-block pair.  Let $\mathcal A$ be the
smallest such factor.  The vertex intervals $W$ form a constant number of
families of pairwise disjoint vertex sets.
Together with \eqref{eq:overview-clipped-energy}, this gives
\begin{equation}
    \sum_W\sum_K E_{K,W}(x)
    \le O(\log n)Q_{\TT_n}^+(x).
    \label{eq:overview-interval-energy-sum}
\end{equation}
Consequently, a bound on $\mathbb E\mathcal A$ gives the desired uniform
relative-error bound with only a logarithmic loss.

To bound $\mathcal A$, normalize the local energy to one.  Paths through
the intervening vertices then constrain the endpoint magnitudes: the
vector with smaller support has Euclidean norm $O(D^{-1/2})$
(\cref{lem:local-completion,lem:completion-order}).
This need not hold for both vectors, so we also use the internal arcs.
In the original vertex order their constraint is
\[
    \underbrace{\sum_{\substack{i<i'\\i,i'\in I}}
        (\alpha_i-\alpha_{i'})_+^2}_{\text{decreases of }\alpha}
    +
    \underbrace{\sum_{\substack{j<j'\\j,j'\in J}}
        (\beta_{j'}-\beta_j)_+^2}_{\text{increases of }\beta}
    \le 1.
\]
These are one-sided constraints: for example, a decreasing $\beta$
contributes zero to the second sum, regardless of its variation.
We therefore need an averaging argument that uses only the controlled
direction (\cref{lem:one-sided-averaging}).

Suppose $\alpha$ has no larger support than $\beta$ and set
$h=X^\top\alpha$, so that the bilinear form is $h^\top\beta$.
Partition $J$ into $R$
consecutive vertex subblocks.  Compare negative coefficients in an
earlier subblock with positive coefficients in a later one, using the
average $\bar\beta$ on a middle subblock.  The basic identity, for
$h_i=-v\le0$ in the earlier block and $h_j=u\ge0$ in the later block, is
\[
    u\beta_j-v\beta_i
    =\underbrace{u(\beta_j-\bar\beta)+v(\bar\beta-\beta_i)}
        _{\text{controlled differences}}
      +\underbrace{(u-v)\bar\beta}_{\text{coefficient imbalance}}.
\]
The first two deviations are bounded above by averages of increases of
$\beta$, precisely the differences controlled by the internal energy.
After summing over subblocks, the last term records the imbalance
between their positive and negative coefficient totals.  The
probabilistic part of the proof bounds these imbalances uniformly over
$\alpha$.  Applying the comparison to both $h$ and $-h$ bounds the
absolute bilinear form, with both endpoint vectors allowed to depend on
the realized matrix.  If $\beta$ has smaller support, we fix $\beta$
instead and use the analogous comparison for decreases of $\alpha$.

Increasing $R$ makes the subblocks shorter.  It increases the bound on
the comparison errors and coefficient fluctuations, but leaves fewer
coordinates in the constant number of boundary subblocks without a
comparison partner.  The resulting tradeoff is
\[
    \mathbb E\mathcal A
    \le \widetilde O\!\left(
        \underbrace{\sqrt{R/\rho}}_{\text{comparisons and fluctuations}}
        +\underbrace{1/R}_{\text{boundary subblocks}}
    \right).
\]
Balancing at $R\asymp\rho^{1/3}$ gives
$\mathbb E\mathcal A\le\widetilde O(\rho^{-1/3})$
(\cref{thm:local-quenched}).  The proof supplies this estimate for
$\rho\ge C\log^{3/2}(2n)$, using the full constraints from the intervening
paths as well as the internal arcs.  Applying the bound on the sum of
clipped energies in \eqref{eq:overview-interval-energy-sum} gives
\begin{equation}
    \label{eq:overview-distance-bound}
    \mathbb E_{\zeta,g}
    \sup_{Q_{\TT_n}^+(x)>0}
    \left|
      \sum_{e=(i,j)\in E_D}
      g_e\frac{\zeta_e}{p_D}
      \frac{(x_i-x_j)_+^2}{Q_{\TT_n}^+(x)}
    \right|
    \le \log^{O(1)}(2n)\rho^{-1/3}.
\end{equation}
Finally, summing over the $O(\log n)$ distance ranges and taking
$\rho=\varepsilon^{-3}\log^{O(1)}(2n)$ makes the expected uniform error
a sufficiently small multiple of $\varepsilon$.  The sampling criterion in
\cref{thm:qgld-general} then gives a sparsifier with
$\widetilde O(n/\varepsilon^3)$ arcs.
This explains the exponent $3$ in this upper bound.  The uniform
Gaussian-width bound \eqref{eq:overview-distance-bound} will also control
the signed transitive-tournament terms used below for arbitrary tournaments.

\paragraph{From a transitive order to an arbitrary tournament.}
The preceding argument uses a fixed transitive orientation.  For an arbitrary
tournament $T$, an arbitrary vertex order would not provide enough directed
paths without overusing some arcs.  We instead order the vertices by
nondecreasing outdegree and let $R$ be the transitive tournament whose arcs
point from later to earlier vertices in this order.  If an arc $u\to v$ of
$R$ has the reverse orientation $v\to u$ in $T$, the outdegree comparison
guarantees a directed path $u\to z\to v$ in $T$.  We split one unit of flow
equally among all such paths; arcs shared by $R$ and $T$ are routed directly.
Every arc of $T$ then carries $O(\log n)$ units of flow in total.  To prove
this, we bound its load by sums of $1/(d_S^+(z)+1)$ or $1/(d_S^-(z)+1)$ in
appropriate subtournaments $S$ of $T$, where $d_S^+$ and $d_S^-$ denote
outdegree and indegree in $S$.  For any subtournament $S$, both sums are
$O(\log n)$ (\cref{lem:tournament-reciprocal-degrees}).  Applying
\eqref{eq:overview-two-arc} along the routed paths gives
\begin{equation}
    \label{eq:overview-order-domination}
    Q_R^+(x)\le O(\log n)Q_T^+(x)
    \qquad\forall x\in\mathbb R^{V(T)}.
\end{equation}

The bound \eqref{eq:overview-order-domination} lets us compare an error
measured relative to $Q_R^+$ with the target energy $Q_T^+$, but we must
still account for pairs on which the
two orientations disagree.  The identity
\begin{equation}
    \label{eq:overview-orientation-reversal}
    (a-b)_+^2=(a-b)^2-(b-a)_+^2
\end{equation}
expresses the contribution in one orientation as an undirected squared
difference minus the contribution in the other orientation.  Applying
\eqref{eq:overview-orientation-reversal} to every disagreement pair gives
\begin{equation}
    \label{eq:overview-tournament-decomposition}
    Q_T^+(x)
    =\sum_{e\in E(R)}\sigma_e q_e^R(x)+K_{\mathcal F}(x),
\end{equation}
where $q_e^R(x)=(x_u-x_v)_+^2$ for $e=(u,v)\in E(R)$, and $\sigma_e$ is
$+1$ when $T$ and $R$ agree on the pair and $-1$ otherwise.  The graph
$\mathcal F$ has one undirected edge for each disagreement pair, and
$K_{\mathcal F}(x)=\sum_{\{u,v\}\in E(\mathcal F)}(x_u-x_v)^2$.
By \eqref{eq:overview-orientation-reversal}, each term in $K_{\mathcal F}$
is the sum of the directed contributions in $T$ and $R$ on that pair.
Consequently, \eqref{eq:overview-order-domination} gives
\begin{equation}
    \label{eq:overview-correction-bound}
    K_{\mathcal F}(x)\le Q_T^+(x)+Q_R^+(x)
    \le O(\log n)Q_T^+(x).
\end{equation}

To sample the signed sum in \eqref{eq:overview-tournament-decomposition},
we extend \eqref{eq:overview-distance-bound} to allow each summand to be
multiplied by the fixed sign $\sigma_e$.  The bound is unchanged because
multiplying independent standard Gaussian variables by fixed signs preserves
their joint distribution.  Summing these bounds over the distance ranges and
applying \cref{lem:sampling-gaussian-comparison} controls the expected
uniform sampling error of the signed sum, normalized by $Q_R^+(x)$
(\cref{lem:signed-tt-sampling}).

To approximate the undirected correction $K_{\mathcal F}$ for all potentials,
we use undirected spectral sampling \cite{SpielmanSrivastava2011}.
In this method, called \emph{effective-resistance sampling}, the edge
inclusion probabilities are determined by the effective resistances of the
edges of $\mathcal F$, defined in \eqref{eq:undirected-effective-resistance}.
We use the independent sampling guarantee in
\cref{lem:independent-spectral-sampling}.

We now construct the sparsifier directly as a subgraph of $T$.  For each
unordered vertex pair, let $p$ be the larger of the endpoint-distance
sampling probability for $R$ and the probability chosen for $\mathcal F$
by effective-resistance sampling.  The probability for $\mathcal F$ is zero
when the pair is not an edge of $\mathcal F$.  Independently for each pair,
retain its arc of $T$ with probability $p$ and give it weight $1/p$ if
retained.

To analyze this subgraph, apply \eqref{eq:overview-tournament-decomposition}
using the pairs just retained.  A retained pair contributes its term in the
signed sum, divided by $p$, and, if it is an edge of $\mathcal F$, its term
in $K_{\mathcal F}$, also divided by $p$.  An omitted pair contributes zero
to both sums.  The identity \eqref{eq:overview-orientation-reversal} thus
continues to hold for each pair, so the two contributions add to exactly
the energy of the retained arc of $T$ with its positive weight $1/p$.

The endpoint-distance probabilities control the error in the signed sum
relative to $Q_R^+$, and the effective-resistance probabilities control
the error in the undirected energy relative to $K_{\mathcal F}$.
Choosing the larger probability preserves both error bounds
(\cref{lem:probability-increase,lem:independent-spectral-sampling}).
Both $Q_R^+$ and $K_{\mathcal F}$ are at most $O(\log n)Q_T^+$ by
\eqref{eq:overview-order-domination} and \eqref{eq:overview-correction-bound}.
Using a sufficiently small constant times $\varepsilon/\log(2n)$ as the
error parameter for each approximation therefore gives a
$(1\pm\varepsilon)$-spectral sparsifier of $T$ with
$\widetilde O(n/\varepsilon^3)$ arcs
(\cref{thm:tournament-sparsification}).

\subsection{Relation to previous work}

This subsection relates our sampling arguments and sparsification bounds to
previous work on undirected spectral sampling, ReLU losses, nonlinear
directed Laplacians, and directed cut sparsification.

Ene, Miller, Pachocki, and Sidford introduced cut balance in their study of
routing and maximum flow in directed graphs~\cite{EneEtAl2016}.
Ikeda and Tanigawa constructed directed cut sparsifiers whose number of arcs
depends linearly on cut balance~\cite{IkedaTanigawa2018}.
Cen et al. further studied the dependence on cut balance for both for-all
and for-each cut sparsification~\cite{CenEtAl2021}.  Their for-all theorem gives
$\widetilde O(\beta n/\varepsilon^2)$ arcs for a digraph with cut balance at
most $\beta$.  They also proposed preserving $Q_G^+$ as a spectral extension
of directed cut sparsification and asked whether balanced digraphs admit sparse
approximations for this energy~\cite[Section~7]{CenEtAl2021}.  Our Eulerian
lower bound rules out a nearly linear general guarantee even when the input is
simple, unweighted, and has cut balance one.  It does not characterize the
minimum sparsifier size for weighted or $\beta$-balanced digraphs in general.

Li, Wang, Yang, and Zhang used Lewis-weight sampling to construct sparsifiers
for the unsquared energy
$\sum_{(u,v)\in E(G)}w_{uv}(x_u-x_v)_+$~\cite{LiWangYangZhang2020}.
For a strongly connected digraph with cut balance at most $\beta$, their
sparsifier has $\widetilde O(\beta^2 n/\varepsilon^2)$ arcs and approximates
this energy within a factor of $1\pm\varepsilon$ for every real potential.
The square in $Q_G^+$ is therefore essential to our cut--spectral separation.

Effective-resistance sampling preserves an undirected quadratic form
\cite{SpielmanSrivastava2011}.  It cannot be applied directly to $Q_G^+$,
because the arcs contributing to this energy depend on the potential, but it
does apply to the undirected correction $K_{\mathcal F}$ in
\eqref{eq:overview-tournament-decomposition}.  The $p=2$ case of the
sensitivity-sampling theorem of Munteanu and Omlor for losses of the form
$t\mapsto(t)_+^p$~\cite{MunteanuOmlor2024} applies to
directed incidence matrices, but its size depends on $\mu(A)$.  Our Eulerian
upper bound supplies the missing bound on this parameter by removing short
directed cycles.  This step uses the
theorem of Huang, Ma, Shapira, Sudakov, and Yuster that a simple Eulerian
digraph with $n$ vertices and $m$ arcs contains a directed cycle of length
$O(n^2/m)$~\cite{HuangMaShapiraSudakovYuster2013}.

Yoshida introduced the nonlinear Laplacian for digraphs, whose associated
energy is $Q_G^+$, and proved a Cheeger inequality for it
\cite{Yoshida2016}.  He subsequently introduced Laplacians for submodular
transformations, unifying the nonlinear Laplacians and Cheeger inequalities
for directed graphs and hypergraphs~\cite{Yoshida2019}.
Fujii, Soma, and Yoshida gave polynomial-time algorithms for the resulting
Laplacian systems and identified their dual optimization problems, which
specialize to quadratic nonnegative-flow problems on directed
graphs~\cite{FujiiSomaYoshida2021}.  This duality underlies
\cref{cor:flow-cost-preservation}.
The energy $Q_G^+$ is also called \emph{directed variation}
in graph signal processing: it has been used to define frequencies and
Fourier transforms for directed graphs, including applications to signal
denoising~\cite{ShafipourEtAl2019}, and as an objective for learning sparse
representations of graph signals~\cite{ShafipourMateos2018}.

Soma and Yoshida defined spectral sparsification
for the corresponding directed-hypergraph energy~\cite{SomaYoshida2019}.
Kapralov, Krauthgamer, Tardos, and Yoshida obtained nearly quadratic-size
spectral sparsifiers when every hyperarc has constant
size~\cite{KapralovEtAl2021STOC}, and Oko, Sakaue, and Tanigawa obtained
nearly tight quadratic-size sparsifiers for general directed
hypergraphs~\cite{OkoSakaueTanigawa2023}.  The complete
bipartite DAG shows why a quadratic worst-case bound is unavoidable without a
restriction such as the tournament or Eulerian conditions considered here.
For comparison, Kapralov et al. proved that undirected hypergraphs admit
spectral sparsifiers with $\widetilde O(n)$ hyperedges for every fixed error
parameter, independently of the maximum hyperedge
size~\cite{KapralovEtAl2021FOCS}.
Lee~\cite{Lee2023} and, independently, Jambulapati, Liu, and
Sidford~\cite{JambulapatiLiuSidford2023} subsequently obtained sparsifiers
with $O(n\varepsilon^{-2}\log n\log r)$ hyperedges, where $r$ is the maximum
hyperedge size.  Both works use chaining to bound the supremum of the
sampling error over all potentials.

The phrase \emph{directed Laplacian sparsification} is also used for a
distinct, matrix-based line of work.  Cohen et al. introduced an asymmetric
spectral approximation for directed Laplacians~\cite{CohenEtAl2017}.  For an
Eulerian digraph, their definition controls the difference between the
nonsymmetric Laplacian matrices relative to the undirected Laplacian obtained
by symmetrization.  They constructed
$\widetilde O(n\varepsilon^{-2})$-arc sparsifiers in almost-linear time and
used them in algorithms for directed Laplacian systems and random-walk
problems.  Chu et al. obtained
$O(n\varepsilon^{-2}\log^4 n)$-arc Eulerian sparsifiers using short cycle
decompositions~\cite{ChuEtAl2018}, and subsequent work improved this approach
using discrepancy methods~\cite{SachdevaThudiZhao2024}.  Jambulapati et al.
later gave a near-linear-time algorithm that, for polynomially
bounded weights, produces Eulerian sparsifiers with
$O(n\varepsilon^{-2}\log^2 n\log^2\log n)$ arcs
\cite{JambulapatiEtAl2025}.  These results compare linear operators; they do
not require pointwise preservation of the nonlinear energy $Q_G^+$.  Their
nearly linear bounds therefore do not conflict with the
$\Omega(n^2)$ lower bound witnessed by the complete bipartite DAG.

\subsection{Organization}

We conclude the introduction by describing where each result is proved.

\Cref{sec:general} defines the directed energy, proves the quadratic-flow
interpretation in \cref{cor:flow-cost-preservation}, and develops the sampling
criterion.
\Cref{sec:simple-eulerian} then proves the upper and lower bounds for simple
Eulerian digraphs, establishing the cut--spectral separation first.
\Cref{sec:transitive-tournaments} constructs a sparsifier of a transitive
tournament by choosing inclusion probabilities from endpoint distances and
proving the required Gaussian-width bound.  \Cref{sec:arbitrary-tournaments}
extends the sampling guarantee to sums
$\sum_e\sigma_e q_e^R$ for fixed coefficients
$\sigma_e\in\{-1,+1\}$ and combines this sampling with undirected spectral
sampling to prove \cref{thm:tournament-sparsification}.
The appendix contains the Gaussian-matrix estimates used above.

\section{Directed Energy and Independent Sampling}
\label{sec:general}

This section fixes the sparsification model used throughout the paper and
derives its interpretation in terms of quadratic costs of nonnegative flows.
It then gives a condition under which independent Bernoulli sampling produces
a sparsifier.

\subsection{The energy of the nonlinear directed Laplacian and its sparsifiers}

This subsection defines the directed energy and the minimum support size of
a spectral sparsifier.  It then represents relative error as a linear
functional of normalized arc-energy vectors.

Let $G=(V,E,w)$ be a finite directed graph with positive arc weights.  For an
arc $e=(u,v)$ and a potential $x\in\mathbb R^V$, define
\[
    q_e^+(x):=w_e(x_u-x_v)_+^2,
    \qquad
    Q_G^+(x):=\sum_{e\in E}q_e^+(x).
\]
We refer to $Q_G^+$ as the energy of the nonlinear directed Laplacian.
A nonnegatively reweighted subgraph $H\subseteq G$ is a
\emph{$(1\pm\varepsilon)$-spectral sparsifier} if
\[
    (1-\varepsilon)Q_G^+(x)
    \le Q_H^+(x)
    \le(1+\varepsilon)Q_G^+(x)
    \qquad\forall x\in\mathbb R^V.
\]
For a coefficient vector $c\in\mathbb R_{\ge0}^E$, define
\[
    Q_{G,c}^+(x):=\sum_{e\in E}c_eq_e^+(x).
\]
We write
\[
    \operatorname{spn}_\varepsilon(G)
    :=
    \min\left\{
       |\supp c|:
       c\in\mathbb R_{\ge0}^E,
       |Q_{G,c}^+(x)-Q_G^+(x)|
       \le\varepsilon Q_G^+(x)
       \quad\forall x\in\mathbb R^V
    \right\}.
\]

Let
\[
    \Delta_E:=\left\{r\in\mathbb R_{\ge0}^E:\sum_{e\in E}r_e=1\right\}
\]
be the probability simplex on $E$.  Every potential $x$
with $Q_G^+(x)>0$ determines the normalized arc-energy vector
\[
    r_x(e):=\frac{q_e^+(x)}{Q_G^+(x)}.
\]
Define
\[
    \mathcal R_G:=\{r_x:Q_G^+(x)>0\}\subseteq\Delta_E,
    \qquad
    \mathcal R_G^0:=\mathcal R_G\cup\{0\}.
\]
The coefficient vector $c$ defines a $(1\pm\varepsilon)$-spectral sparsifier
precisely when
\[
    \sup_{r\in\mathcal R_G}|\langle c-\mathbf 1,r\rangle|
    \le\varepsilon.
\]

\subsection{Quadratic costs of nonnegative flows}
\label{subsec:flow-costs}

This subsection proves that a multiplicative approximation of the directed
energy preserves the optimum quadratic cost of routing every feasible demand
with nonnegative arc flows.  We use the duality between directed energies and
quadratic flow costs studied by Fujii, Soma, and
Yoshida~\cite{FujiiSomaYoshida2021}, and give a direct proof in our notation.

Let $B_G\in\mathbb R^{V\times E}$ be the vertex--arc incidence matrix whose
column for an arc $e=(u,v)$ is $\mathbf 1_u-\mathbf 1_v$.  For a demand vector
$b\in\mathbb R^V$, let $\mathcal F_G(b)$ denote the minimum quadratic cost of
a nonnegative flow routing $b$:
\[
    \mathcal F_G(b)
    =
    \min_{\substack{f\in\mathbb R_{\ge0}^E\\B_Gf=b}}
       \sum_{e\in E}\frac{f_e^2}{w_e},
\]
where the minimum is $+\infty$ when $b$ is not routable by a nonnegative flow.

\begin{proposition}[Dual formula for quadratic flow cost]
\label{prop:flow-cost-duality}
For every $b\in\mathbb R^V$,
\begin{equation}
    \label{eq:flow-cost-duality}
    \mathcal F_G(b)
    =
    \sup_{x\in\mathbb R^V}
       \left\{2\langle b,x\rangle-Q_G^+(x)\right\}.
\end{equation}
In particular, the supremum is $+\infty$ exactly when $b$ is infeasible.
\end{proposition}

\begin{proof}
Suppose first that $b$ is feasible.  Introduce a multiplier $2x$ for the
constraint $B_Gf=b$.  The infimum of the Lagrangian over $f\ge0$ separates
over the arcs.  For $e=(u,v)$, its contribution is determined by
\[
    \sup_{f_e\ge0}
       \left\{2(x_u-x_v)f_e-\frac{f_e^2}{w_e}\right\}
    =w_e(x_u-x_v)_+^2.
\]
Finite-dimensional convex quadratic programming duality now gives
\eqref{eq:flow-cost-duality}.

If $b$ is infeasible, then $b$ does not belong to the closed polyhedral cone
$\{B_Gf:f\ge0\}$.  By separation, there is an $x\in\mathbb R^V$ such that
$\langle b,x\rangle>0$ and $B_G^\top x\le0$.  Hence $Q_G^+(x)=0$, and replacing
$x$ by a positive multiple makes the supremum in
\eqref{eq:flow-cost-duality} infinite.
\end{proof}

\begin{proof}[Proof of \cref{cor:flow-cost-preservation}]
For every $a>0$, the two-homogeneity of $Q_G^+$ and
\cref{prop:flow-cost-duality} give
\begin{equation}
    \label{eq:scaled-flow-dual}
    \sup_x\left\{2\langle b,x\rangle-aQ_G^+(x)\right\}
    =\frac{1}{a}\mathcal F_G(b).
\end{equation}
The inequality $Q_H^+\le(1+\varepsilon)Q_G^+$ and
\eqref{eq:scaled-flow-dual} imply
\[
    \mathcal F_H(b)
    \ge \frac{1}{1+\varepsilon}\mathcal F_G(b).
\]
Similarly, $Q_H^+\ge(1-\varepsilon)Q_G^+$ implies
\[
    \mathcal F_H(b)
    \le \frac{1}{1-\varepsilon}\mathcal F_G(b).
\]
Because $b$ is feasible in $G$, the last upper bound is finite.  Thus $b$ is
also feasible in $H$.
\end{proof}

\subsection{Expected Gaussian width after random restriction}

This subsection defines the Gaussian width that controls the expected uniform
error of independent Bernoulli sampling.  We also record the sampling
comparison used both here and for signed sums in the tournament argument.

Suppose that each arc $e$ is retained independently with probability $\pi_e\in(0,1]$
and reweighted by $1/\pi_e$.  Let $(\xi_e)_{e\in E}$ and
$(\xi'_e)_{e\in E}$ be two mutually independent families with
$\xi_e,\xi'_e\sim\operatorname{Bernoulli}(\pi_e)$, and define
\[
    \zeta_e:=\mathbf 1\{\xi_e\ne\xi'_e\}.
\]
We call $\zeta=(\zeta_e)_{e\in E}$ the vector of disagreement indicators.
Define
\begin{equation}
    \mathfrak G_G(\pi)
    :=\mathbb E_{\zeta,g}\sup_{r\in\mathcal R_G^0}
    \left|\sum_eg_e\frac{\zeta_e}{\pi_e}r_e\right|,
    \label{eq:restricted-gaussian-width}
\end{equation}
where $(g_e)_{e\in E}$ is a family of independent standard Gaussian random
variables, independent of $\zeta$.  Thus we first take the Gaussian width after
restricting to the coordinates with $\zeta_e=1$ and then average over the
restriction.

\begin{lemma}[Sampling comparison]
\label{lem:sampling-gaussian-comparison}
Let $E$ be a finite set, let $\mathcal F\subseteq\mathbb R^E$ be bounded, and
let $\pi\in(0,1]^E$.  Let $(\xi_e)_e$ and $(\xi'_e)_e$ be independent families
with $\xi_e,\xi'_e\sim\operatorname{Bernoulli}(\pi_e)$, put
$\zeta_e=\mathbf 1\{\xi_e\ne\xi'_e\}$, and let $(g_e)_e$ be independent
standard Gaussian variables.  Then
\begin{equation}
  \mathbb E_\xi\sup_{f\in\mathcal F}
  \left|\sum_e\left(\frac{\xi_e}{\pi_e}-1\right)f_e\right|
  \le
  C\mathbb E_{\zeta,g}\sup_{f\in\mathcal F\cup\{0\}}
  \left|\sum_e g_e\frac{\zeta_e}{\pi_e}f_e\right|.
  \label{eq:sampling-gaussian-comparison}
\end{equation}
\end{lemma}

\begin{proof}
Symmetrization bounds the left-hand side of
\eqref{eq:sampling-gaussian-comparison} by
\[
  \mathbb E_{\xi,\xi'}\sup_{f\in\mathcal F\cup\{0\}}
  \left|\sum_e\frac{\xi_e-\xi'_e}{\pi_e}f_e\right|.
\]
Conditional on $\zeta$, write
$\xi_e-\xi'_e=\tau_e\zeta_e$, where the nonzero $\tau_e$ are independent
Rademacher variables.  The Rademacher--Gaussian comparison
inequality~\cite[Eq.~(4.8)]{LedouxTalagrand1991} bounds the conditional
expectation by the Gaussian supremum on the right-hand side of
\eqref{eq:sampling-gaussian-comparison}, up to a universal constant.
\end{proof}

\subsection{Independent sampling from an expected Gaussian-width bound}

The next theorem converts an upper bound on the expected Gaussian width
$\mathfrak G_G(\pi)$ in \eqref{eq:restricted-gaussian-width} into a spectral
sparsifier whose support size is proportional to the total inclusion probability.

\begin{theorem}[Sparsification from an expected Gaussian-width bound]
\label{thm:qgld-general}
There are universal constants $c_{\mathrm{samp}},C>0$ such that, for every
finite directed graph $G=(V,E,w)$ with positive arc weights and every
$0<\varepsilon<1$, the following holds.  If
$\pi\in(0,1]^E$ satisfies
\[
    \mathfrak G_G(\pi)\le c_{\mathrm{samp}}\varepsilon,
\]
then there exists a reweighted subgraph $H\subseteq G$ satisfying
\[
    (1-\varepsilon)Q_G^+(x)
    \le Q_H^+(x)
    \le(1+\varepsilon)Q_G^+(x)
    \qquad\forall x\in\mathbb R^V,
\]
and
\[
    |E(H)|\le C\sum_{e\in E}\pi_e.
\]
\end{theorem}

\begin{proof}
Fix $\pi$ as in the theorem, and let
$(\xi_e)_{e\in E}$ be independent variables with
$\xi_e\sim\operatorname{Bernoulli}(\pi_e)$.  Put
\[
    Z(\xi)
    :=
    \sup_{r\in\mathcal R_G}
    \left|\sum_e\left(\frac{\xi_e}{\pi_e}-1\right)r_e\right|.
\]
Applying \cref{lem:sampling-gaussian-comparison} with
$\mathcal F=\mathcal R_G$ yields
\[
    \mathbb E Z(\xi)
    \le C\mathfrak G_G(\pi)
    \le Cc_{\mathrm{samp}}\varepsilon.
\]
Choose $c_{\mathrm{samp}}$ so that
$C c_{\mathrm{samp}}\varepsilon\le\varepsilon/8$.  Markov's
inequality then gives
\[
    \mathbb P[Z(\xi)>\varepsilon]\le\frac18.
\]
The number $N=\sum_e\xi_e$ of retained arcs has mean $\sum_e\pi_e$, and hence
\[
    \mathbb P\left[N>4\sum_e\pi_e\right]\le\frac14.
\]
With positive probability, both $Z(\xi)\le\varepsilon$ and
$N\le4\sum_e\pi_e$ hold.  Fixing such a realization proves the theorem.
\end{proof}

\section{\texorpdfstring{The $n^{3/2}$ Threshold for Simple Eulerian
Digraphs}{The n to the 3/2 Threshold for Simple Eulerian Digraphs}}
\label{sec:simple-eulerian}

Throughout this section, every digraph is simple.  A digraph is
\emph{Eulerian} if the indegree and outdegree agree at every vertex.
We show that although Eulerian digraphs have cut balance one, their
worst-case sparsifier size for
$Q_G^+$ is $\widetilde\Theta(\min\{m,n^{3/2}\})$ for constant error.  We prove
the upper bound by removing short directed cycles and the lower bound using
cyclic blow-ups.

\subsection{Upper bound by removing short directed cycles}
\label{subsec:eulerian-arc-count}

We first state the row-sampling theorem used for the upper bound and show that
a decomposition into short directed cycles satisfies its hypothesis.  We
then remove cycles until the residual is sparse and apply the theorem to the
removed arcs.

For a matrix $A\in\mathbb R^{N\times d}$ with rows
$a_1^\top,\ldots,a_N^\top$, define
\[
  f_A(x):=\sum_{i=1}^N(a_i^\top x)_+^2=\|(Ax)_+\|_2^2
\]
and
\begin{equation}
  \mu(A):=
  \sup_{Ax\ne0}\frac{\|Ax\|_2^2}{\|(Ax)_+\|_2^2}.
  \label{eq:relu-mu-definition}
\end{equation}
Following Munteanu and Omlor, we call $\mu(A)$ the $\mu$-complexity of $A$.

\begin{lemma}[Sensitivity sampling for squared ReLU losses~\cite{MunteanuOmlor2024}]
\label{lem:relu-row-sampling}
Let $A\in\mathbb R^{N\times d}$ have rank $r$ and finite $\mu(A)$.  For every
$0<\eta\le1/2$, there is a weight vector $w\in\mathbb R_{\ge0}^N$ with
\[
  \left|
    \sum_{i=1}^N w_i(a_i^\top x)_+^2
    -\sum_{i=1}^N(a_i^\top x)_+^2
  \right|
  \le\eta\sum_{i=1}^N(a_i^\top x)_+^2
  \qquad\forall x\in\mathbb R^d
\]
and
\[
  |\supp w|
  \le\widetilde O\left(\frac{r\mu(A)}{\eta^2}\right).
\]
The hidden factors are polylogarithmic in $r$, $\mu(A)$, and $\eta^{-1}$.
\end{lemma}

\begin{proof}
Put $\eta_0:=\min\{\eta,1/4\}$.  The case $p=2$ and $h(t)=(t)_+^2$ of
\cite[Theorem~G.1]{MunteanuOmlor2024}, applied with error $\eta_0$ and
constant failure probability, independently retains rows and assigns
inverse-probability weights.  Its guarantee with error $\eta_0\le\eta$ is
stronger than the required one, while $\eta_0\ge\eta/2$ keeps the support
bound within a universal constant factor.  If $A$ does not have full column
rank, replace it by a full-column-rank matrix with the same column space; this
preserves both the function class and $\mu(A)$.
\end{proof}

For a digraph $F$, let $A_F\in\mathbb R^{E(F)\times V(F)}$ be its
arc--vertex incidence matrix, with row
$\mathbf 1_u^\top-\mathbf 1_v^\top$ for the arc $u\to v$.  For a scalar $t$,
write $(t)_-:=\max\{-t,0\}$, and apply $(\cdot)_+$ and $(\cdot)_-$
coordinatewise to vectors.  Then
\[
  Q_F^+(x)=\|(A_Fx)_+\|_2^2.
\]

\begin{lemma}[Bounded $\mu$-complexity from short cycles]
\label{lem:bounded-cycle-asymmetry}
If the arc set of $F$ is the disjoint union of directed cycles of length at
most $L$, then $\mu(A_F)\le L$.
\end{lemma}

\begin{proof}
On a directed cycle of length $\ell$, put
$d_i=x_{v_i}-x_{v_{i+1}}$, with indices modulo $\ell$.  Since
$\sum_i d_i=0$,
\[
  \sum_i(d_i)_-^2
  \le\left(\sum_i(d_i)_-\right)^2
  =\left(\sum_i(d_i)_+\right)^2
  \le(\ell-1)\sum_i(d_i)_+^2.
\]
Summing over the cycles gives
$\|(A_Fx)_-\|_2^2\le(L-1)\|(A_Fx)_+\|_2^2$, and hence
$\|A_Fx\|_2^2\le L\|(A_Fx)_+\|_2^2$.
\end{proof}

\begin{theorem}[Upper bound for simple Eulerian digraphs]
\label{thm:simple-eulerian-upper-bound}
Let $G$ be a simple Eulerian digraph with $n$ vertices and $m$ arcs.  For every
$0<\varepsilon\le1/2$,
\[
  \operatorname{spn}_{\varepsilon}(G)
  \le
  \widetilde O\left(
    \min\left\{
      m,\frac{n^{3/2}}{\varepsilon}+\frac n{\varepsilon^2}
    \right\}
  \right).
\]
\end{theorem}

\begin{proof}
Fix $1\le M\le m$.  Starting with $G$, repeatedly remove a directed cycle
while the current residual has more than $M$ arcs.  Removing a directed cycle
preserves the Eulerian property.  If a residual has $m_t$ arcs, one of its
strongly connected components $C$ satisfies
$|E(C)|/|V(C)|^2\ge m_t/n^2$: an Eulerian digraph has no arcs between distinct
strongly connected components.  The theorem of Huang, Ma, Shapira, Sudakov,
and Yuster~\cite[Corollary~1.2]{HuangMaShapiraSudakovYuster2013} therefore gives a
directed cycle of length at most $6n^2/m_t<6n^2/M$.

Let $R$ be the final residual and let $F$ contain the removed arcs.  Then
$|E(R)|\le M$, and the arcs of $F$ are partitioned into directed cycles of
length at most
\[
  L_M:=\left\lceil\frac{6n^2}{M}\right\rceil.
\]
By \cref{lem:bounded-cycle-asymmetry},
$\mu(A_F)\le L_M$, while $\operatorname{rank}(A_F)\le n-1$.
\Cref{lem:relu-row-sampling} gives a $(1\pm\varepsilon)$-spectral sparsifier
$H_F$ of $F$ with $\widetilde O(nL_M/\varepsilon^2)$ arcs.  Retain $R$
without changing its weights and put $H=R\cup H_F$.  Since
$Q_G^+=Q_R^++Q_F^+$, the graph $H$ approximates $Q_G^+$ within
$1\pm\varepsilon$ and
\begin{equation}
  |E(H)|
  \le
  M+\widetilde O\left(
    \frac n{\varepsilon^2}+\frac{n^3}{M\varepsilon^2}
  \right).
  \label{eq:simple-eulerian-tradeoff}
\end{equation}
If $m$ is below a constant multiple of $n^{3/2}/\varepsilon$, retain every
arc.  Otherwise take $M=\Theta(n^{3/2}/\varepsilon)$ in
\eqref{eq:simple-eulerian-tradeoff}.
\end{proof}

\subsection{Lower bound via cyclic blow-ups}

We construct Eulerian digraphs whose spectral sparsifiers must retain most
arcs for constant error.  The witness potentials vary an arbitrary subset of
the arcs from one fixed vertex while keeping all other energy bounded.

Fix integers $k\ge4$ and $b\ge1$, partition the vertices into
\[
  V=V_0\mathbin{\dot\cup}\cdots\mathbin{\dot\cup}V_{k-1},
  \qquad |V_i|=b,
\]
and include all arcs
\[
  V_i\longrightarrow V_{i+1\bmod k}
  \qquad(i\in\mathbb Z_k).
\]
Denote the resulting digraph by $G^{\mathrm{cyc}}_{k,b}$.  It is simple,
strongly connected, and Eulerian: every vertex has $b$ incoming and $b$
outgoing arcs.  It has $n=kb$ vertices and $m=kb^2$ arcs.

By cyclic symmetry, consider the arcs from $V_0$ to $V_1$.  For
$u\in V_0$ and $T\subseteq V_1$, define $x^{u,T}$ by
\begin{align}
  x^{u,T}_u&=1,
  &x^{u,T}_v&=0 &&(v\in V_0\setminus\{u\}),\notag\\
  x^{u,T}_v&=0 &&(v\in T),
  &x^{u,T}_v&=1 &&(v\in V_1\setminus T),\notag\\
  x^{u,T}_v&=\frac{k-r}{k-2}
  &&(v\in V_r,\ 2\le r\le k-1).
  \label{eq:unweighted-witness}
\end{align}
Put
\begin{equation}
  \Gamma_{k,b}
  :=\frac{(k-3)b^2+b(b-1)}{(k-2)^2}
  =\frac{b^2}{k-2}-\frac{b}{(k-2)^2}.
  \label{eq:cyclic-baseline}
\end{equation}

\begin{theorem}[Support lower bound for cyclic blow-ups]
\label{thm:unweighted-eulerian-lower-bound}
Let $0<\varepsilon\le1/2$, and let $H$ be a $(1\pm\varepsilon)$-spectral
sparsifier of $G^{\mathrm{cyc}}_{k,b}$.  Then
\begin{equation}
  |E(H)|
  \ge
  kb\left(
    b-\frac{2\varepsilon}{1-\varepsilon}\Gamma_{k,b}
  \right).
  \label{eq:unweighted-general-lower-bound}
\end{equation}
In particular, if $k=8b+2$, then
\[
  |E(H)|\ge\frac34kb^2=\Omega(n^{3/2}),
  \qquad n=kb.
\]
\end{theorem}

\begin{proof}
For the potential $x^{u,T}$, precisely the arcs $u\to v$ with $v\in T$
contribute between $V_0$ and $V_1$, and the arcs from $V_1$ to $V_2$
contribute zero.  Each of the $k-3$ pairs from $V_2\to V_3$ through
$V_{k-2}\to V_{k-1}$ contributes $b^2/(k-2)^2$, and the arcs from
$V_{k-1}$ to $V_0$ contribute $b(b-1)/(k-2)^2$.  Hence
\begin{equation}
  Q^+_{G^{\mathrm{cyc}}_{k,b}}(x^{u,T})
  =\Gamma_{k,b}+|T|.
  \label{eq:cyclic-witness-energy}
\end{equation}

Fix a layer $V_i$ and $u\in V_i$; by cyclic relabeling take $i=0$.  Let
\[
  Z_u:=\{v\in V_1:\text{the weight of }u\to v\text{ in }H\text{ is zero}\}.
\]
The only arc contributions that change between $x^{u,\varnothing}$ and
$x^{u,Z_u}$ are those of $u\to v$ with $v\in Z_u$, and $H$ assigns weight
zero to all of them.  Therefore
\begin{equation}
  Q_H^+(x^{u,Z_u})=Q_H^+(x^{u,\varnothing}).
  \label{eq:unweighted-equal-output-energy}
\end{equation}
Combining \eqref{eq:cyclic-witness-energy},
\eqref{eq:unweighted-equal-output-energy}, and the approximation inequalities in
\eqref{eq:intro-sparsifier-guarantee} gives
\[
  (1-\varepsilon)(\Gamma_{k,b}+|Z_u|)
  \le(1+\varepsilon)\Gamma_{k,b}.
\]
Thus
\[
  |Z_u|\le\frac{2\varepsilon}{1-\varepsilon}\Gamma_{k,b}.
\]
By cyclically relabeling the layers, the same bound applies to the number of
omitted arcs from any $u\in V_i$ to $V_{i+1\bmod k}$.  Thus every vertex
retains at least
$b-2\varepsilon\Gamma_{k,b}/(1-\varepsilon)$ outgoing arcs, which proves
\eqref{eq:unweighted-general-lower-bound}.

If $k=8b+2$, then $\Gamma_{k,b}<b/8$ and
$2\varepsilon/(1-\varepsilon)\le2$, so every vertex retains at least $3b/4$
arcs.  Since $n=(8b+2)b\le10b^2$, this is $\Omega(n^{3/2})$.
\end{proof}

The upper bound in \cref{thm:simple-eulerian-upper-bound} and the lower bound
for cyclic blow-ups in \eqref{eq:unweighted-general-lower-bound} determine
the worst-case sparsifier size for every value of $m$.

\begin{proof}[Proof of \cref{thm:eulerian-n32-transition}]
The upper bound follows from \cref{thm:simple-eulerian-upper-bound}.
For the lower bound, put $M_\star:=\min\{m,n^{3/2}\}$.

If $M_\star\le32n$, let $\ell=\lfloor\min\{m,n\}\rfloor$.  For sufficiently
large $n$, we have $\ell\ge3$ and $\ell=\Omega(M_\star)$.  Take a directed
cycle of length $\ell$ and pad it with isolated vertices.  Every cycle arc is
indispensable: for a target arc $v_0\to v_1$, assign values one and zero to
$v_0$ and $v_1$, respectively, and increase the potential monotonically from
zero to one along the complementary directed path.  The target is the only
arc with a positive difference.

Suppose now that $M_\star>32n$, and put
\[
  b:=\left\lfloor\frac{M_\star}{16n}\right\rfloor,
  \qquad
  k:=\left\lfloor\frac nb\right\rfloor.
\]
Then
\[
  \frac{M_\star}{32n}\le b\le\frac{\sqrt n}{16},
\]
so $k\ge8b+2$ for sufficiently large $n$.  The graph
$G_0:=G^{\mathrm{cyc}}_{k,b}$, padded with isolated vertices, has at most $n$
vertices, and
\[
  |E(G_0)|=kb^2\le nb\le\frac{M_\star}{16}\le m.
\]
Moreover,
\[
  |E(G_0)|=kb^2\ge(n-b)b\ge\frac{M_\star}{64}.
\]
Since $k\ge8b+2$, \eqref{eq:cyclic-baseline} gives
$\Gamma_{k,b}<b^2/(k-2)\le b/8$.  Applying
\eqref{eq:unweighted-general-lower-bound} with
$2\varepsilon_0/(1-\varepsilon_0)\le2$ shows that every sparsifier of $G_0$
retains at least $3|E(G_0)|/4=\Omega(M_\star)$ arcs.
\end{proof}

\section{Sampling Transitive Tournaments by Endpoint Distance}
\label{sec:transitive-tournaments}

This section constructs spectral sparsifiers of transitive tournaments.
For a relative-error parameter $0<\varepsilon\le1/2$, we show that every
$n$-vertex transitive tournament has a $(1\pm\varepsilon)$-spectral sparsifier
with $\widetilde O(\min\{n^2,n/\varepsilon^3\})$ arcs.

Throughout this section, $C,c>0$ denote universal constants that may
change from line to line.

\Cref{sec:distance-scale-sampling} states an expected Gaussian-width bound
for arcs in one range of endpoint distances (\cref{thm:scale-response}).
Assuming this bound, \cref{sec:tt-sparsifier} chooses inclusion probabilities
for all distance ranges and applies the general sampling criterion in
\cref{thm:qgld-general} to obtain the sparsifier.

To prove the bound for one distance range, \cref{sec:value-tree}
decomposes squared differences into bilinear terms and assigns arcs to
consecutive vertex intervals.  \Cref{sec:scale-assembly} derives energy
constraints for the endpoint vectors and reduces the Gaussian-width bound
to an estimate for their bilinear forms, completing the sum over intervals
under that estimate.  \Cref{sec:prob-tools} proves the estimate in three
steps: a deterministic averaging inequality, elimination of the endpoint
vectors, and moment bounds for the remaining random quantities.

\subsection{A Gaussian-width bound for one endpoint-distance range}
\label{sec:distance-scale-sampling}

This subsection isolates the arcs whose endpoint distances lie in
$[D,2D)$ and samples each arc independently with probability $p=\rho/D$,
where $\rho\in[1,D]$ controls the expected sample size.  We define the
associated Gaussian process and bound its expected supremum by a
polylogarithmic factor times $\rho^{-1/3}$; this estimate will determine the
inclusion probability for the distance range.

For a positive integer $n$, let $\TT_n$ be the tournament on $[n]$ with the
arc $i\to j$ whenever $i<j$.  We call $j-i$ the endpoint distance of this arc.

For every positive integer $D$, define
\[
    E_D:=\{(i,j):D\le j-i<2D\}.
\]
Fix $\rho$ with $1\le\rho\le D$ and put $p:=\rho/D$, so that
$1/D\le p\le1$.  Let $(\xi_e)_{e\in E_D}$ and
$(\xi'_e)_{e\in E_D}$ be two mutually independent families of
Bernoulli$(p)$ random variables, and set
$\zeta_e:=\mathbf 1\{\xi_e\ne\xi'_e\}$.
Specialize $\mathfrak G_G(\pi)$ from \eqref{eq:restricted-gaussian-width}
by taking
$G=\TT_n$, $\pi_e=p$ for $e\in E_D$, and $\pi_e=1$ otherwise.
Arcs outside $E_D$ then have zero disagreement indicators and contribute
nothing.  We denote the resulting expected Gaussian width by
\begin{equation}
    \mathfrak G_{n,D}(\rho)
    :=
    \mathbb E_{\zeta,g}
    \sup_{Q_{\TT_n}^+(x)>0}
    \left|
       \sum_{e\in E_D}
       g_e\frac{\zeta_e}{p}
       \frac{q_e^+(x)}{Q_{\TT_n}^+(x)}
    \right|,
    \label{eq:tt-distance-gaussian-width}
\end{equation}
where $(g_e)_{e\in E_D}$ is a family of independent standard Gaussian random
variables, independent of the two Bernoulli families $(\xi_e)_{e\in E_D}$
and $(\xi'_e)_{e\in E_D}$.
Here the first subscript is the number of vertices, and the second is the
distance parameter in $E_D$.  The inequalities $1\le\rho\le D$ ensure that
$p=\rho/D$ is a valid inclusion probability.  Since $|E_D|\le nD$, each of the
two Bernoulli samples retains at most $n\rho$ arcs in expectation.

The next theorem bounds $\mathfrak G_{n,D}(\rho)$.

\begin{theorem}[Bound on $\mathfrak G_{n,D}(\rho)$]
\label{thm:scale-response}
There are universal constants
$A_0,B_0,\kappa_0>0$ such that, for every integer $D\ge1$ and every
\[
    A_0\log^{3/2}(2n)\le\rho\le D,
\]
one has
\[
    \mathfrak G_{n,D}(\rho)
    \le
    B_0\log^{\kappa_0}(2n)\rho^{-1/3}.
\]
\end{theorem}

We defer the proof of \cref{thm:scale-response} to
\cref{sec:scale-assembly}.

\subsection{Combining the distance ranges to obtain a sparsifier}
\label{sec:tt-sparsifier}

We now use \cref{thm:scale-response} to construct a spectral sparsifier of
$\TT_n$.  For each dyadic distance range $[D,2D)$ and inclusion probability
$\rho/D$, that theorem gives
$\mathfrak G_{n,D}(\rho)\le B_0\log^{\kappa_0}(2n)\rho^{-1/3}$ when
$A_0\log^{3/2}(2n)\le\rho\le D$.
Let $0<\eta\le1/2$ be a target upper bound on the expected Gaussian width
$\mathfrak G_{\TT_n}(\pi)$ for the entire tournament.  The following lemma
achieves this target with total inclusion probability
$\widetilde O(\min\{n^2,n/\eta^3\})$.  Taking $\eta$ to be a sufficiently
small multiple of the desired relative error $\varepsilon$ then gives a
sparsifier by \cref{thm:qgld-general}.

\begin{lemma}[Sampling probabilities for a transitive tournament]
\label{lem:tt-sampling-probabilities}
For every $0<\eta\le1/2$, there are inclusion probabilities
$\pi\in(0,1]^{E(\TT_n)}$ such that arcs with the same endpoint distance
receive the same probability,
\begin{equation}
  \sum_{e\in E(\TT_n)}\pi_e
  \le \widetilde O\!\left(
    \min\left\{n^2,\frac{n}{\eta^3}\right\}
  \right)
  \label{eq:tt-sampling-budget}
\end{equation}
and
\begin{equation}
  \mathfrak G_{\TT_n}(\pi)\le\eta.
  \label{eq:tt-sampling-width}
\end{equation}
\end{lemma}

\begin{proof}
Let $\mathcal D$ be the set of dyadic integers $D$ for which $E_D$ is
nonempty.  Let
$C_{\mathrm{sum}}\ge1$ be a universal constant, and put
\[
  \eta_0:=\frac{\eta}{C_{\mathrm{sum}}\log(2n)}.
\]
Choose
\[
  \rho_\star:=C_{\rho}\log^{3\kappa_0}(2n)\eta_0^{-3},
\]
where the universal constant $C_{\rho}$ is large enough that
$\rho_\star\ge A_0\log^{3/2}(2n)$ and
$B_0\log^{\kappa_0}(2n)\rho_\star^{-1/3}\le\eta_0$.
For every $D\in\mathcal D$ and every $e\in E_D$, set
\[
  \pi_e:=\min\left\{1,\frac{\rho_\star}{D}\right\}.
\]

All arcs in $E_D$ are retained when $D\le\rho_\star$, so their disagreement
indicators are zero.  The triangle inequality and
\cref{thm:scale-response}, applied with $\rho=\rho_\star$ on each range
with $D>\rho_\star$, give
\[
  \mathfrak G_{\TT_n}(\pi)
  \le\sum_{\substack{D\in\mathcal D\\D>\rho_\star}}
       \mathfrak G_{n,D}(\rho_\star)
  \le |\mathcal D|\eta_0
  \le\eta,
\]
where the last inequality uses $|\mathcal D|=O(\log(2n))$ and a sufficiently
large choice of $C_{\mathrm{sum}}$.

There are $O(nD)$ arcs in $E_D$.  Summing over the values
$D\in\mathcal D$ with $D\le\rho_\star$ contributes $O(n\rho_\star)$
probabilities, while every $D\in\mathcal D$ with $D>\rho_\star$ contributes
\[
  |E_D|\frac{\rho_\star}{D}=O(n\rho_\star).
\]
Thus
\[
  \sum_e\pi_e
  \le \min\left\{\binom n2,Cn\rho_\star\log(2n)\right\}
  =\widetilde O\!\left(
    \min\left\{n^2,\frac{n}{\eta^3}\right\}
  \right),
\]
which proves \eqref{eq:tt-sampling-budget}.
\end{proof}

\begin{corollary}[Sparsification of transitive tournaments]
\label{cor:tt-sparsification}
For every $0<\varepsilon\le1/2$, the transitive tournament $\TT_n$ has a
$(1\pm\varepsilon)$-spectral sparsifier with
$\widetilde O(\min\{n^2,n/\varepsilon^3\})$ arcs.
\end{corollary}

\begin{proof}
Apply \cref{lem:tt-sampling-probabilities} with
$\eta=\min\{c_{\mathrm{samp}},1\}\varepsilon$, where $c_{\mathrm{samp}}$ is
the constant in \cref{thm:qgld-general}.  Then
$\mathfrak G_{\TT_n}(\pi)\le c_{\mathrm{samp}}\varepsilon$, so
\cref{thm:qgld-general} gives a sparsifier with
$O(\sum_e\pi_e)=\widetilde O(\min\{n^2,n/\varepsilon^3\})$ arcs.
\end{proof}

\subsection{Decomposing squared differences and grouping arcs into vertex intervals}
\label{sec:value-tree}

For the arcs $E_D=\{(i,j):D\le j-i<2D\}$ of the $n$-vertex transitive
tournament $\TT_n$, we decompose the random sum
\begin{equation}
    S_D(x):=\sum_{(i,j)\in E_D}\frac{\zeta_{ij}g_{ij}}p(x_i-x_j)_+^2.
    \label{eq:tt-distance-random-sum}
\end{equation}
Here $p=\rho/D$, and $\zeta_{ij},g_{ij}$ are the random variables from
\eqref{eq:tt-distance-gaussian-width}.  Dividing
\eqref{eq:tt-distance-random-sum} by $Q_{\TT_n}^+(x)$ gives the expression
inside the absolute value in that definition.
We use two decompositions, with different roles.  A node $K=[\ell_K,r_K]$
of a tree on the attained values will index a bilinear term in each
squared difference.  A consecutive vertex interval $W\subseteq[n]$ will
contain the endpoints of a group of arcs, together with $\Theta(D)$
intervening vertices.  Thus $K$ is an interval of values, whereas $W$ is
an interval in the fixed vertex order.

For such a value interval $K$, clip each potential value to $K$ and define
the energy on a vertex set $W$ by
\[
    x_i^K:=\min\{r_K,\max\{\ell_K,x_i\}\},
    \qquad E_{K,W}(x):=Q_{\TT_n[W]}^+(x^K).
\]
Our goal is to express $S_D(x)$ as a sum of local bilinear contributions
$Z_{K,W}(x)$ and obtain bounds of the form
\begin{equation}
    |Z_{K,W}(x)|\le C\mathcal A_W E_{K,W}(x),
    \qquad
    \sum_{W,K}E_{K,W}(x)\le C\log(2n)Q_{\TT_n}^+(x).
    \label{eq:tt-local-accounting-plan}
\end{equation}
Here $\mathcal A_W$ will be a random supremum over normalized endpoint
vectors, defined in \eqref{eq:endpoint-bilinear-supremum}, that does not
depend on $x$ or its value tree.  Consequently, a bound on
$\mathbb E\max_W\mathcal A_W$ will control the entire distance range.
The value tree below gives the exact decomposition and the logarithmic
energy bound; the vertex intervals then let us sum these energies with
only constant overlap.

\subsubsection{A binary tree on the attained values}

We express each squared positive difference $(x_i-x_j)_+^2$ as a sum of
products, so that \eqref{eq:tt-distance-random-sum} can be written as a sum
of bilinear forms.

Fix a nonconstant potential $x$ and list its distinct attained values as
\[
    y_0<y_1<\cdots<y_m.
\]
For $1\le r\le m$, let $J_r=(y_{r-1},y_r)$ be the consecutive value gaps.
Use a balanced binary tree with $m$ ordered leaves and height $O(\log(2m))$,
with its branching structure fixed in advance for each $m$.  Associate the
$r$th leaf with the closed interval $[y_{r-1},y_r]$.  Every node $K$ represents
the interval spanned by its descendant leaves.  For an internal node $K$, let
$K^-$ and $K^+$ be the intervals represented by its lower and upper children,
respectively, and let $t_K$ be their common endpoint.  Define
\[
    \alpha_K(a):=\operatorname{length}(K^+\cap(-\infty,a)),
\]
\[
    \beta_K(b):=\operatorname{length}(K^-\cap(b,\infty)).
\]

\begin{lemma}[Exact decomposition of a squared difference]
\label{lem:value-tree-square}
For every two values $a>b$ attained by $x$,
\[
    (a-b)^2
    =
    \sum_{J_r\subset(b,a)}|J_r|^2
    +
    2\sum_{K\text{ internal}}\alpha_K(a)\beta_K(b).
\]
\end{lemma}

\begin{proof}
Partition the square $(b,a)\times(b,a)$ according to the value gaps
containing its two coordinates.  The squares where both coordinates lie
in the same gap give the first sum.  Group the remaining rectangles by the
lowest common ancestor $K$ of their two distinct gaps.  The higher gap lies
in $K^+$ and the lower gap in $K^-$.  For each $K$, the rectangles in this
group, including both coordinate orders, have total area
$2\alpha_K(a)\beta_K(b)$.  Summing over $K$ gives the second sum.
\end{proof}

For the clipped energies $E_{K,W}(x)=Q_{\TT_n[W]}^+(x^K)$ defined above,
the next lemma gives the logarithmic bound needed in
\eqref{eq:tt-local-accounting-plan}.

\begin{lemma}[Sum of clipped energies over the tree]
\label{lem:subtree-accounting}
For every $W\subseteq[n]$,
\[
    \sum_{K\text{ a tree node}}E_{K,W}(x)
    \le
    C\log(2n)Q_{\TT_n[W]}^+(x).
\]
More generally, for every family $\mathcal W$ of pairwise disjoint subsets of
$[n]$,
\begin{equation}
    \sum_{W\in\mathcal W}\sum_{K\text{ a tree node}}E_{K,W}(x)
    \le
    C\log(2n)Q_{\TT_n}^+(x).
    \label{eq:disjoint-clipped-energy-sum}
\end{equation}
\end{lemma}

\begin{proof}
Fix an arc whose endpoint values are $a>b$.  At any fixed depth, the tree
nodes represent intervals with disjoint interiors.  If
\[
    d_K:=\operatorname{length}(K\cap(b,a)),
\]
then the contribution of this arc to the clipped energy at node $K$ is
$d_K^2$, and therefore
\[
    \sum_{K\text{ at this depth}}d_K^2
    \le
    \left(\sum_{K\text{ at this depth}}d_K\right)^2
    \le
    (a-b)^2.
\]
The tree has depth $O(\log n)$, so summing over its depths and then over the
arcs of $\TT_n[W]$ proves the first assertion.  The statement for a family
$\mathcal W$ follows by summing the energies of the arc-disjoint induced
subtournaments $\TT_n[W]$, $W\in\mathcal W$.
\end{proof}

\subsubsection{Consecutive vertex intervals for
\texorpdfstring{$E_D$}{E-D}}
\label{sec:topological-windows}

We assign each arc of $E_D$ to a consecutive vertex interval containing
its endpoints.  We will arrange that these intervals form a constant
number of families of pairwise vertex-disjoint intervals, so that their
energies can be summed using \cref{lem:subtree-accounting}.

Assume $32\le D<n$ and put $\ell=\lfloor D/16\rfloor$.  Partition $[n]$
into consecutive blocks of length $\ell$; if a final block has fewer than
$\ell$ vertices, merge it into the preceding block.  Denote the resulting
blocks by $B_1,B_2,\ldots$.  Every block has between $\ell$ and $2\ell$
vertices, so its size lies between $D/32$ and $D/8$.

For every pair $B_s,B_{s+r}$ joined by an arc of $E_D$, define
\[
    W_{s,r}:=B_s\cup B_{s+1}\cup\cdots\cup B_{s+r}.
\]
There are $(r-1)\ell$ or more vertices strictly between the endpoint
blocks.  Since some arc joining them has distance less than $2D$ and
$\ell>D/24$, we have $1\le r\le48$.
Write $\mathcal W_D$ for the family of these intervals.  An interval
$W\in\mathcal W_D$ uniquely determines its first and last blocks, which we
denote by $I_W,J_W$.  Set
\[
    M_W:=W\setminus(I_W\cup J_W),
    \qquad
    \mathcal S_W:=E_D\cap(I_W\times J_W).
\]
The sets $I_W,M_W,J_W$ occur in this order.  The endpoint sizes are
$\Theta(D)$, and $|M_W|=\Theta(D)$: for an arc $(i,j)\in\mathcal S_W$,
the difference between $j-i$ and $|M_W|$ is at most
$|I_W|+|J_W|\le D/4$.  Also $|\mathcal W_D|\le48n$, and the sets
$\mathcal S_W$ partition $E_D$.

For each offset $r$, group the intervals $W_{s,r}$ according to $s$ modulo
$r+1$.  Intervals in the same group are vertex-disjoint.  This gives at
most $\sum_{r=1}^{48}(r+1)=1224$ groups, so the bound on the sum of clipped
energies in \eqref{eq:disjoint-clipped-energy-sum} applies within each group.

\subsection{Reducing the Gaussian-width bound to endpoint bilinear forms}
\label{sec:scale-assembly}

We reduce the Gaussian-width estimate to a uniform bound for bilinear forms
on the endpoint blocks $I_W,J_W$.  The target is the local estimate
$|Z_{K,W}(x)|\le C\mathcal A_WE_{K,W}(x)$ from
\eqref{eq:tt-local-accounting-plan}.  First, the energy on $W$ constrains
the endpoint vectors of each clipped potential.  These constraints define
the normalized supremum $\mathcal A_W$.  We then state its expected bound
and sum the local contributions to prove \cref{thm:scale-response}.

\subsubsection{Energy constraints on the endpoint vectors}
\label{sec:local-completion}

Arcs within an endpoint block constrain ordered differences of its vector.
Paths through the intervening vertices impose an additional quadratic
constraint.  We express both constraints below; they will let us normalize
each bilinear contribution by the clipped energy on $W$.

For a nonnegative vector $z\in\mathbb R^m$ and $s\ge1$, define
\[
    V_m(z):=\sum_{i<j}(z_i-z_j)_+^2,
    \qquad
    \Psi_s(z):=\inf_{\gamma\ge0}
       \left\{\sum_i(z_i-\gamma)_+^2+s\gamma^2\right\}.
\]
For $D\ge1$, put
\begin{equation}
    \cE_{\to,s,D}(z):=V_m(z)+D\Psi_s(z),
    \qquad
    \cE_{\leftarrow,s,D}(z):=V_m^{\rm rev}(z)+D\Psi_s(z),
    \label{eq:ordered-vector-energy}
\end{equation}
where $V_m^{\rm rev}(z)=\sum_{i<j}(z_j-z_i)_+^2$.
In the application, $m$ is the endpoint-block size and $s$ is the support
size on the opposite block.  The term $D\Psi_s$ accounts for paths through
the $\Theta(D)$ intervening vertices, with $\gamma$ measuring a deviation
of an intervening potential value from the split point.

Fix $W\in\mathcal W_D$ and a node $K$ of the value tree for $x$.
For an internal node, let $t_K$ be its split point; for a leaf, let $t_K$
be the midpoint of its value gap.  Using increasing vertex order on each
endpoint block, put
\begin{equation}
    \alpha_i=(x_i^K-t_K)_+\quad(i\in I_W),
    \qquad
    \beta_j=(t_K-x_j^K)_+\quad(j\in J_W),
    \label{eq:endpoint-vectors}
\end{equation}
and write $a=|\supp\alpha|$, $b=|\supp\beta|$.
For an internal node these vectors agree with
$\alpha_K(x_i)$ and $\beta_K(x_j)$ in the squared-difference decomposition
of \cref{lem:value-tree-square}.

\begin{lemma}[Energy constraints on the endpoint vectors]
\label{lem:local-completion}
If $a,b\ge1$, then
\[
    E_{K,W}(x)\ge
    c\bigl[\cE_{\to,b,D}(\alpha)+
           \cE_{\leftarrow,a,D}(\beta)\bigr]
\]
for a universal constant $c>0$.
\end{lemma}

\begin{proof}
Arcs within $I_W$ and $J_W$ contribute at least $V_{|I_W|}(\alpha)$ and
$V_{|J_W|}^{\rm rev}(\beta)$, respectively.
For $k\in M_W$ with $x_k^K\ge t_K$, put $\gamma=x_k^K-t_K$.
The arcs from $I_W$ to $k$ and from $k$ to $J_W$ contribute at least
\[
    \sum_i(\alpha_i-\gamma)_+^2+b\gamma^2+\|\beta\|_2^2
    \ge \Psi_b(\alpha)+\Psi_a(\beta),
\]
where $\Psi_a(\beta)\le\|\beta\|_2^2$ follows by choosing threshold zero.
If $x_k^K<t_K$, the same argument with the roles of $\alpha,\beta$
interchanged gives this lower bound.  Summing over the $\Theta(D)$
vertices of $M_W$ proves the assertion.  This argument applies to both
internal nodes and leaves.
\end{proof}

We will use the following consequence of the threshold constraint to
bound the norm of whichever endpoint vector has smaller support.

\begin{lemma}[Norm bound from the threshold constraint]
\label{lem:completion-order}
If a nonnegative vector $z$ has at most $s$ nonzero coordinates, where
$s\ge1$, then
$\Psi_s(z)\ge\frac12\|z\|_2^2$.
\end{lemma}

\begin{proof}
For every $\gamma\ge0$,
$\|z\|_2\le\|(z-\gamma\one)_+\|_2+\sqrt{s}\gamma$.
Square this inequality, use $(u+v)^2\le2(u^2+v^2)$, and minimize over
$\gamma$.
\end{proof}

In particular, if $a\le b$ and $\cE_{\to,b,D}(\alpha)\le1$, then
\begin{equation}
    \|\alpha\|_2\le\sqrt{2/D},
    \qquad
    \|\alpha\|_1\le\sqrt{2a/D}.
    \label{eq:smaller-support-norms}
\end{equation}
The case $b<a$ gives the analogous bounds for $\beta$.
These estimates control the vector with smaller support; they need not
hold for the other vector.  The averaging argument in
\cref{lem:one-sided-averaging} will use the ordered differences and threshold
constraint of that other vector.

\subsubsection{A uniform bound for the endpoint bilinear forms}
\label{sec:local-theorem}

Let $p=\rho/D$ be the inclusion probability, and use the disagreement
indicators $\zeta_{ij}$ and Gaussian variables $g_{ij}$ from
\eqref{eq:tt-distance-gaussian-width}.  The $\zeta_{ij}$ are independent
Bernoulli$(2p(1-p))$ variables, and the $g_{ij}$ are independent standard
Gaussian variables, independent of $\zeta$.
For each interval $W\in\mathcal W_D$, let $X^W$ be the matrix on
$I_W\times J_W$ with entries $\zeta_{ij}g_{ij}/p$ on $\mathcal S_W$
and zero elsewhere.
In the supremum below, $a=|\supp\alpha|$ and $b=|\supp\beta|$:
\begin{equation}
    \mathcal A_W:=
    \sup_{\substack{
       \alpha\in\mathbb R_{\ge0}^{I_W},\
       \beta\in\mathbb R_{\ge0}^{J_W}\\
       a,b\ge1\\
       \cE_{\to,b,D}(\alpha)\le1,\
       \cE_{\leftarrow,a,D}(\beta)\le1}}
       |\alpha^\top X^W\beta|.
    \label{eq:endpoint-bilinear-supremum}
\end{equation}
Thus $\mathcal A_W$ optimizes over both vectors, including their support
sets, after the matrix is realized.  The next theorem bounds the expected
maximum over the intervals as well.

\begin{theorem}[Uniform bound for the endpoint bilinear forms]
\label{thm:local-quenched}
There are universal constants $B_1,C_1,\kappa_1>0$ such that, for all
integers $32\le D<n$ and all
$B_1\log^{3/2}(2n)\le\rho\le D$,
\[
    \mathbb E\max_{W\in\mathcal W_D}\mathcal A_W
    \le C_1\log^{\kappa_1}(2n)\rho^{-1/3}.
\]
\end{theorem}

We prove \cref{thm:local-quenched} in \cref{sec:prob-tools}.

\subsubsection{Summing the clipped contributions}

By homogeneity, the bilinear supremum $\mathcal A_W$ bounds the contribution
of each tree node by a constant times $\mathcal A_WE_{K,W}(x)$.
The squared-difference decomposition in \cref{lem:value-tree-square} and
the clipped-energy sum in \eqref{eq:disjoint-clipped-energy-sum} now give
the bound for the entire distance range.

\begin{proof}[Proof of \cref{thm:scale-response}]
If $D\ge n$, then $E_D$ is empty.  Otherwise, choose $A_0\ge B_1$ large
enough that $\rho\ge A_0\log^{3/2}(2n)$ and $D\ge\rho$ imply $D\ge32$.
Use the vertex intervals $\mathcal W_D$ constructed in
\cref{sec:topological-windows}.

Fix the random matrices $X^W$ and a potential $x$ of positive energy.
For each tree node $K$ and interval $W$, use the vectors $\alpha,\beta$
in \eqref{eq:endpoint-vectors} and set
\[
    Z_{K,W}(x):=
    \begin{cases}
       2\alpha^\top X^W\beta,&K\text{ is internal},\\
       4\alpha^\top X^W\beta,&K\text{ is a leaf}.
    \end{cases}
\]
The internal-node coefficient is from \cref{lem:value-tree-square}.
For a leaf gap of length $\Delta$, the nonzero coordinates of $\alpha$
and $\beta$ equal $\Delta/2$, so $4\alpha_i\beta_j=\Delta^2$ exactly when
the endpoint values straddle that gap.
Since the sets $\mathcal S_W$ partition $E_D$, the decomposition gives
\begin{equation}
    \sum_{W\in\mathcal W_D}\sum_K Z_{K,W}(x)
    =\sum_{(i,j)\in E_D}\frac{\zeta_{ij}g_{ij}}p(x_i-x_j)_+^2.
    \label{eq:interval-sum-is-gaussian-process}
\end{equation}

If either vector has empty support, then $Z_{K,W}(x)=0$.
Otherwise, put $u=\cE_{\to,b,D}(\alpha)$ and
$v=\cE_{\leftarrow,a,D}(\beta)$, where $a,b$ are their support sizes.
These quantities are positive.  Normalizing the two vectors by
$\sqrt u,\sqrt v$ in \eqref{eq:endpoint-bilinear-supremum} and using
\cref{lem:local-completion} gives
\[
    |Z_{K,W}(x)|
    \le4\mathcal A_W\sqrt{uv}
    \le2\mathcal A_W(u+v)
    \le C\mathcal A_WE_{K,W}(x).
\]
The intervals form at most $1224$ families of pairwise disjoint sets.
Summing the clipped energies within each family by
\eqref{eq:disjoint-clipped-energy-sum}, we obtain
\[
    \sum_W\sum_K|Z_{K,W}(x)|
    \le C\log(2n)\left(\max_W\mathcal A_W\right)Q_{\TT_n}^+(x).
\]
This realizes the two bounds in \eqref{eq:tt-local-accounting-plan}.
The value tree may depend on $x$, but $\max_W\mathcal A_W$ does not.
Dividing by $Q_{\TT_n}^+(x)$, taking the supremum over potentials of
positive energy, and then taking expectation therefore gives
\begin{equation}
    \mathfrak G_{n,D}(\rho)
    \le C\log(2n)\,\mathbb E\max_W\mathcal A_W.
    \label{eq:tt-width-reduction}
\end{equation}
\Cref{thm:local-quenched} proves the desired bound after
choosing $\kappa_0\ge\kappa_1+1$ and $B_0$ sufficiently large.
\end{proof}

\subsection{Bounding the endpoint bilinear forms}
\label{sec:prob-tools}

It remains to bound $\mathbb E\max_W\mathcal A_W$ in
\eqref{eq:tt-width-reduction}.  We will divide an endpoint vector into $R$
consecutive subblocks of length $L\asymp D/R$.  Keeping $R$ free will give
\begin{equation}
    \|\mathcal A_W\|_{L_q}
    \le C\sqrt q\left(\sqrt{R/\rho}+1/R\right),
    \qquad q:=\lceil\log(2n)\rceil,
    \label{eq:bilinear-R-tradeoff}
\end{equation}
for integers $R\ge4$ satisfying \eqref{eq:subblock-parameter-conditions}
below.  Here
$\|Y\|_{L_q}:=(\mathbb E|Y|^q)^{1/q}$; this logarithmic moment will also
control the maximum over $W$.  The first term comes from comparing
coordinates across subblocks and from fluctuations in their coefficient
totals.  The second comes from the constant number of subblocks left
without a comparison partner.  Balancing them at the end will give
$R\asymp\rho^{1/3}$ and the desired $\rho^{-1/3}$ bound.

We first prove a deterministic averaging inequality.  Applying the
endpoint constraints will then remove the vectors and their support
sizes, leaving three kinds of random quantities to estimate.

\subsubsection{A deterministic averaging inequality}
\label{sec:one-sided-averaging}

In the application, suppose that $\alpha$ has the smaller support, and
reverse $\beta$ and the matching columns of the endpoint matrix.  Write
$X$ for the matrix in this column order and $z\in\mathbb R^k$ for the
reversed vector.  In this order the energy controls
$V_k(z)=\sum_{i<j}(z_i-z_j)_+^2$: increases of $\beta$ in the original
vertex order become decreases of $z$.  For fixed $\alpha$, the bilinear
form is $h^\top z$, where $h=X^\top\alpha$ is its coefficient vector.
The inequality below treats arbitrary $h$ and uses only these controlled
decreases of $z$, not a bound on its total variation.

Let $R\ge4$ and partition $[k]$ into consecutive subblocks
$T_1,\ldots,T_R$, the first $R-1$ of size $L$ and the last of size between
$L$ and $2L$, where $1\le L\le D$.  Fix $h\in\mathbb R^k$ and $z\ge0$.
For $r\le R-3$, use the mean
$c_r:=L^{-1}\sum_{t\in T_{r+1}}z_t$ on the middle subblock to compare a
positive coefficient $h_i=u\ge0$ in $T_r$ with a negative coefficient
$h_j=-v\le0$ in $T_{r+2}$:
\begin{equation}
    uz_i-vz_j
    =\underbrace{u(z_i-c_r)+v(c_r-z_j)}_{\text{controlled differences}}
      +\underbrace{(u-v)c_r}_{\text{coefficient imbalance}}.
    \label{eq:three-block-comparison}
\end{equation}
Both deviations in the first term are bounded above by averages of
forward decreases, as illustrated in \cref{fig:three-block-averaging}.
The last term records the difference between the coefficient magnitudes.

To collect these comparisons over whole subblocks, define the positive
and negative coefficient totals
\[
    P_r(h):=\sum_{j\in T_r}(h_j)_+,
    \qquad N_r(h):=\sum_{j\in T_r}(-h_j)_+.
\]
\begin{figure}[htbp]
  \centering
  \begin{tikzpicture}[x=1cm,y=1cm,font=\small,>=stealth]
  \draw[->,thick] (0,2.0) -- (14.6,2.0)
    node[midway,above] {reversed vertex order on $J_W$};
  \draw[draw=blue!65!black,fill=blue!4,rounded corners=2pt]
    (0,0) rectangle (4.2,1.4);
  \draw[draw=black!60,fill=black!3,rounded corners=2pt]
    (5.2,0) rectangle (9.4,1.4);
  \draw[draw=blue!65!black,fill=blue!4,rounded corners=2pt]
    (10.4,0) rectangle (14.6,1.4);
  \node at (2.1,1.1) {$T_r$};
  \node at (2.1,0.65) {positive coefficient};
  \node at (2.1,0.23) {$h_i=u$, value $z_i$};
  \node at (7.3,1.1) {$T_{r+1}$};
  \node at (7.3,0.65) {$c_r=L^{-1}\sum_{t\in T_{r+1}}z_t$};
  \node at (7.3,0.23) {comparison value};
  \node at (12.5,1.1) {$T_{r+2}$};
  \node at (12.5,0.65) {negative coefficient};
  \node at (12.5,0.23) {$h_j=-v$, value $z_j$};
  \draw[->,thick,blue!65!black] (2.1,-0.4) -- (7.3,-0.4)
    node[midway,below,align=center]
      {average of $(z_i-z_t)_+$\\bounds $z_i-c_r$};
  \draw[->,thick,blue!65!black] (7.3,-0.4) -- (12.5,-0.4)
    node[midway,below,align=center]
      {average of $(z_t-z_j)_+$\\bounds $c_r-z_j$};
  \node[text=orange!65!black] at (7.3,-1.55)
    {remaining coefficient imbalance: $c_r\bigl(P_r(h)-N_{r+2}(h)\bigr)$};
\end{tikzpicture}
  \caption{Comparison through a middle subblock, in the reversed vertex
  order.  The signs shown are those of the coefficients $h_i,h_j$;
  the vector $z$ is nonnegative and need not be monotone.  Averaging the
  indicated decreases over $t\in T_{r+1}$ bounds the two deviations in
  \eqref{eq:three-block-comparison}.  The coefficient imbalance remains
  to be estimated.}
  \label{fig:three-block-averaging}
\end{figure}
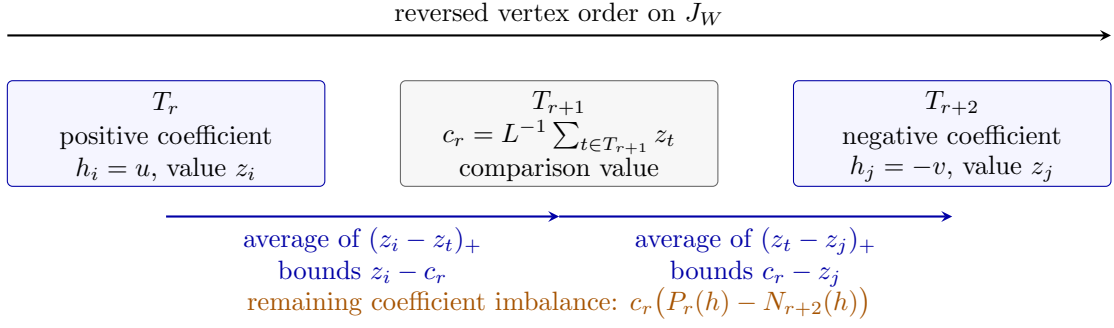

Replacing each coordinate $z_j$ on these two subblocks by $z_j-c_r$
leaves $c_r(P_r(h)-N_{r+2}(h))$ in addition to the difference terms.
We collect the
absolute differences of these totals and the unpaired positive terms in
\[
    U(h):=\sum_{r=1}^{R-3}|P_r(h)-N_{r+2}(h)|
           +P_{R-2}(h)+P_{R-1}(h)+P_R(h).
\]
The last three subblocks occur for a specific reason.  The two-subblock
shift leaves the positive terms on $T_{R-1},T_R$ without partners.  We
also omit the comparison from $T_{R-2}$ to $T_R$, because $T_R$ may be
larger than $L$; the probabilistic argument will compare equal-sized
subblocks.  Any unused negative terms can be discarded in an upper bound
since $z\ge0$.

\begin{lemma}[Averaging under a one-sided energy constraint]
\label{lem:one-sided-averaging}
For the partition above, every $a\ge1$, every $h\in\mathbb R^k$, and
every $z\ge0$ satisfying $V_k(z)+D\Psi_a(z)\le1$ obey
\begin{equation}
    h^\top z\le C\left(\frac{\|h\|_2}{\sqrt L}
                       +\frac{U(h)}{\sqrt{aD}}\right).
    \label{eq:block-average-bound}
\end{equation}
\end{lemma}

\begin{proof}
For $i\in T_r$ and $j\in T_{r+2}$, the middle average satisfies
\[
    z_i-c_r\le\frac1L\sum_{t\in T_{r+1}}(z_i-z_t)_+,
    \qquad
    c_r-z_j\le\frac1L\sum_{t\in T_{r+1}}(z_t-z_j)_+.
\]
Multiply these inequalities by $(h_i)_+$ and $(-h_j)_+$, respectively,
and sum over $r\le R-3$.  Cauchy--Schwarz bounds the difference terms by
$C\|h\|_2\sqrt{V_k(z)/L}$: each adjacent pair of subblocks occurs at most
twice, and each coefficient is repeated $L$ times with a factor $1/L$.
Bounding the coefficient imbalances by their absolute values and
discarding the unused negative terms gives
\begin{equation}
    h^\top z\le C\|h\|_2\sqrt{V_k(z)/L}+w^\top z,
    \label{eq:averaging-difference-remainder}
\end{equation}
where $w\ge0$ puts the constant $|P_r(h)-N_{r+2}(h)|/L$ on $T_{r+1}$
for each $r\le R-3$, and adds $(h_j)_+$ on the last three subblocks.
Thus $\|w\|_1=U(h)$ and $\|w\|_2\le C\|h\|_2$, by Cauchy--Schwarz
on each subblock.

The threshold term $\Psi_a(z)$ controls this remaining nonnegative
linear form.  For every $\gamma\ge0$,
\[
    w^\top z\le C\|h\|_2\|(z-\gamma\one)_+\|_2+\gamma U(h).
\]
Weighted Cauchy--Schwarz and minimization over $\gamma$ give
\[
    w^\top z\le
    C\left(\|h\|_2+\frac{U(h)}{\sqrt a}\right)\sqrt{\Psi_a(z)}
    \le C\left(\frac{\|h\|_2}{\sqrt D}
                   +\frac{U(h)}{\sqrt{aD}}\right).
\]
Combining this estimate with \eqref{eq:averaging-difference-remainder},
$V_k(z)\le1$, and $L\le D$ proves
\eqref{eq:block-average-bound}.
\end{proof}

\subsubsection{Eliminating the endpoint vectors and support sizes}
\label{sec:bilinear-deterministic-reduction}

We now apply \cref{lem:one-sided-averaging} to a bilinear form.  For this
step, keep the partition $T_1,\ldots,T_R$ above, let
$X\in\mathbb R^{m\times k}$ be any matrix, and let
$\alpha\ge0,z\ge0$ satisfy, for some integer $a\ge1$,
\begin{equation}
    \|\alpha\|_2\le\sqrt{2/D},\qquad
    \|\alpha\|_1\le\sqrt{2a/D},\qquad
    V_k(z)+D\Psi_a(z)\le1.
    \label{eq:bilinear-norm-constraints}
\end{equation}
For the normalized endpoint vectors, \eqref{eq:smaller-support-norms}
gives the first two bounds when $a=|\supp\alpha|\le|\supp\beta|$;
reversing $\beta$ turns its energy constraint into the third bound.
We will derive a bound involving only $X$, valid
simultaneously for every such pair and every $a$.

The first term of \eqref{eq:block-average-bound}, with
$h=\pm X^\top\alpha$, is bounded by
$\sqrt{2/D}\,\|X\|_{\rm op}/\sqrt L$.
To handle $U(h)$, define, for $r\le R-3$,
\[
    B_r(X):=\sup_{\substack{u\ge0\\\|u\|_1\le1}}
       \left|\sum_{j\in T_r}(\langle X_{\cdot j},u\rangle)_+
             -\sum_{j\in T_{r+2}}(-\langle X_{\cdot j},u\rangle)_+\right|,
\]
and, for $r=R-2,R-1,R$, define
\[
    S_r(X):=\max_i\sum_{j\in T_r}|X_{ij}|.
\]
The quantities $B_r$ bound the differences of coefficient totals
uniformly over nonnegative vectors of unit $\ell_1$ norm; the quantities
$S_r$ bound the unpaired terms.  Put
\[
    F(X):=\sum_{r=1}^{R-3}B_r(X)+\sum_{r=R-2}^R S_r(X).
\]
Homogeneity and the triangle inequality give
$U(X^\top\alpha)\le\|\alpha\|_1F(X)$.
The dependence on the support size cancels because
\[
    \frac{\|\alpha\|_1}{\sqrt{aD}}
    \le\frac{\sqrt{2a/D}}{\sqrt{aD}}=\frac{\sqrt2}{D}.
\]
Applying \eqref{eq:block-average-bound} to both signs of $X^\top\alpha$
therefore yields
\begin{equation}
    |\alpha^\top Xz|
    \le C\left[\frac{\|X\|_{\rm op}}{\sqrt{DL}}
               +\frac{F(X)+F(-X)}D\right].
    \label{eq:bilinear-matrix-bound}
\end{equation}
The right side contains neither endpoint vector nor $a$.  It remains to
bound the matrix norm, the differences $B_r$, and the boundary sums $S_r$.
No enumeration of support sets or sizes is needed.

\subsubsection{Moment bounds and the choice of the number of subblocks}
\label{sec:bilinear-moment-bounds}

We apply \eqref{eq:bilinear-matrix-bound} to the random endpoint matrices.
After filling their missing entries by a moment comparison, we will have
a matrix $X$ whose entries are independent copies of $\zeta g/p$.
For $L\asymp D/R$ and $Lp\gtrsim q$, the three estimates to be proved are
\[
\begin{array}{lll}
    \text{Contribution} & \text{Random quantity} & L_q\text{ bound}\\[3pt]
    \text{Comparison errors}
      & \|X\|_{\rm op}/\sqrt{DL} & C\sqrt{qR/\rho}\\[3pt]
    \text{Coefficient imbalances}
      & D^{-1}\sum_{r=1}^{R-3}B_r(X) & C\sqrt{qR/\rho}\\[3pt]
    \text{Unpaired subblocks}
      & D^{-1}\sum_{r=R-2}^{R}S_r(X) & C\sqrt q/R.
\end{array}
\]
The same estimates apply to $-X$.  Their sum will give
\eqref{eq:bilinear-R-tradeoff}.

We use two elementary moment estimates, proved in \cref{app:prob-tools}.
For any $t\ge1$, write $\|Y\|_{L_t}=(\mathbb E|Y|^t)^{1/t}$;
a subscript such as $L_t(g\mid\zeta)$ specifies which variables are
averaged and which are held fixed.

\begin{lemma}[Maximum of binomially dominated random variables]
\label{lem:binomial-max}
Let $M$ be a positive integer and $\mu\ge0$.
Let $Z_1,\ldots,Z_M$ be nonnegative random variables, each stochastically
dominated by a binomial random variable with mean at most $\mu$.
Then, for every $t\ge1$,
\[
    \left\|\max_{1\le r\le M}Z_r\right\|_{L_t}
    \le C\bigl(\mu+\log(2M)+t\bigr).
\]
\end{lemma}

\begin{lemma}[Gaussian matrix norm after Bernoulli restriction]
\label{lem:masked-gaussian-norm}
Let $m,k,D,n$ be positive integers with $m,k\le D\le n$, and put
$q=\lceil\log(2n)\rceil$.
For $0<\rho\le D$, put $p=\rho/D$ and let $X\in\mathbb R^{m\times k}$
have entries $X_{ij}=\zeta_{ij}g_{ij}/p$, where the $\zeta_{ij}$ are independent
Bernoulli$(2p(1-p))$ variables and the $g_{ij}$ are independent standard
Gaussian variables, independent of $\zeta$.
If $\rho\ge Cq$, then
\[
    \bigl\|\|X\|_{\rm op}\bigr\|_{L_q}\le C D\sqrt{q/\rho}.
\]
\end{lemma}

\begin{proof}[Proof of \cref{thm:local-quenched}]
Fix $n,D,\rho$ as in the theorem, put $p=\rho/D$, and set
$q=\lceil\log(2n)\rceil\ge2$.  We take $B_1$ sufficiently large below.
It suffices to prove $\|\mathcal A_W\|_{L_q}\le C\sqrt q\,\rho^{-1/3}$
for every $W$, since $|\mathcal W_D|\le48n$ and
\begin{equation}
    \mathbb E\max_W\mathcal A_W
    \le (48n)^{1/q}\max_W\|\mathcal A_W\|_{L_q},
    \label{eq:interval-maximum-from-moments}
\end{equation}
where $(48n)^{1/q}$ is universally bounded.

\paragraph{Preparing the random matrix.}
Fix $W$.  Fill the zero entries of $X^W$ outside $\mathcal S_W$ with
independent copies of $\zeta g/p$, and denote the resulting matrix and
bilinear supremum by $\widehat X^W$ and $\widehat{\mathcal A}_W$.
This makes all entries identically distributed, which will allow us to
compare the coefficient totals on equal-sized subblocks.
The added entries have conditional mean zero.  The $q$th power of the
supremum is convex in the entries, since the feasible vectors form a
deterministic set.  Conditional Jensen gives the moment comparison
\begin{equation}
    \|\mathcal A_W\|_{L_q}
    \le\|\widehat{\mathcal A}_W\|_{L_q}.
    \label{eq:matrix-filling-comparison}
\end{equation}

Consider first feasible pairs with $a=|\supp\alpha|\le|\supp\beta|=b$.
Reverse the columns of $\widehat X^W$, writing $X$ for the resulting
matrix and $z$ for the reversed vector $\beta$.  The entries of $X$ are
still independent copies of $\zeta g/p$, and the endpoint constraints in
\eqref{eq:endpoint-bilinear-supremum} imply
\eqref{eq:bilinear-norm-constraints}.

\paragraph{Keeping the number of subblocks free.}
Choose an integer $R\ge4$ such that
\begin{equation}
    R^2\le D/64,\qquad \rho/R\ge Cq,
    \label{eq:subblock-parameter-conditions}
\end{equation}
where $C$ is a sufficiently large universal constant.  We will check
these conditions for the final choice of $R$ at the end.
Put $k=|J_W|$ and $L=\lfloor k/R\rfloor$.
Since $D/32\le k\le D/8$, we have $L\ge R$ and $L\asymp D/R$.
Partition $[k]$ into $R$ consecutive subblocks, the first $R-1$ of size
$L$.  The last has size $L+(k-RL)$, which lies between $L$ and $2L$.
Also $Lp\asymp\rho/R\ge cCq$.  These are precisely the size and sampling
conditions used in the estimates below.
The deterministic bound \eqref{eq:bilinear-matrix-bound} now applies
simultaneously to all the pairs with $a\le b$, even when they depend on $X$.

\paragraph{The matrix norm: comparison errors.}
By \cref{lem:masked-gaussian-norm},
\begin{equation}
    \left\|\frac{\|X\|_{\rm op}}{\sqrt{DL}}\right\|_{L_q}
    \le \frac{CD\sqrt{q/\rho}}{\sqrt{DL}}
    \le C\sqrt{qR/\rho}.
    \label{eq:comparison-error-moment}
\end{equation}

\paragraph{Differences of coefficient totals.}
For $r\le R-3$, the columns on $T_r$ and the negated columns on
$T_{r+2}$ are independent copies of the same $L$-column matrix.
Exchanging each pair of columns allows independent Rademacher signs to be
inserted into their difference.  Apply the contraction inequality for the
1-Lipschitz function $t\mapsto(t)_+$ to each copy
\cite[Theorem~4.12]{LedouxTalagrand1991}, then use the triangle inequality.
The signs can be absorbed into the symmetric columns of $X$.  Hence
\begin{equation}
    \|B_r(X)\|_{L_q}
    \le C\left\|\max_i\left|\sum_{j\in T_r}X_{ij}\right|\right\|_{L_q}.
    \label{eq:subblock-contraction}
\end{equation}
Here the supremum of $|\langle Y,u\rangle|$ over $u\ge0$,
$\|u\|_1\le1$, is exactly $\|Y\|_\infty$.

Write $X_{ij}=\zeta_{ij}g_{ij}/p$ and let
$\Delta=\max_{i,r}\sum_{j\in T_r}\zeta_{ij}$.
Each count has mean at most $4Lp$, and there are at most $n^2$ counts.
Since $Lp\ge cCq$, \cref{lem:binomial-max} gives
\begin{equation}
    \|\Delta\|_{L_q}\le CLp.
    \label{eq:subblock-retained-count-moment}
\end{equation}
Conditional on $\zeta$, each row sum in \eqref{eq:subblock-contraction}
is Gaussian with variance at most $\Delta/p^2$.
The scalar Gaussian tail bound and a union bound over the rows give
\[
    \left\|\max_i\left|\sum_{j\in T_r}X_{ij}\right|
    \right\|_{L_q(g\mid\zeta)}
    \le C\sqrt{q\Delta}/p.
\]
Averaging over $\zeta$ using \eqref{eq:subblock-retained-count-moment},
applying \eqref{eq:subblock-contraction}, and then summing over $r$ yields
\begin{equation}
\begin{aligned}
    \|B_r(X)\|_{L_q}&\le C\sqrt{qL/p},\\
    \left\|\frac1D\sum_{r=1}^{R-3}B_r(X)\right\|_{L_q}
    &\le\frac{CR}{D}\sqrt{qL/p}
    \le C\sqrt{qR/\rho}.
\end{aligned}
    \label{eq:coefficient-imbalance-moment}
\end{equation}

\paragraph{The unpaired subblocks.}
Let $G=\max_{i,j}|g_{ij}|$.  Pointwise,
$S_r(X)\le\Delta G/p$.
The scalar Gaussian tail bound gives $\|G\|_{L_q}\le C\sqrt q$,
and $G$ is independent of $\Delta$.  Together with the count bound
\eqref{eq:subblock-retained-count-moment}, this gives
\begin{equation}
\begin{aligned}
    \|S_r(X)\|_{L_q}&\le CL\sqrt q,\\
    \left\|\frac1D\sum_{r=R-2}^{R}S_r(X)\right\|_{L_q}
    &\le\frac{CL\sqrt q}{D}\le\frac{C\sqrt q}{R}.
\end{aligned}
    \label{eq:unpaired-subblock-moment}
\end{equation}
There are only three such terms, so this contribution decreases as $R$
increases.

\paragraph{Balancing the bounds.}
Combining the matrix-norm estimate \eqref{eq:comparison-error-moment},
the coefficient-imbalance estimate \eqref{eq:coefficient-imbalance-moment},
and the unpaired-subblock estimate \eqref{eq:unpaired-subblock-moment} gives
\[
    \left\|\frac{\|X\|_{\rm op}}{\sqrt{DL}}
             +\frac{F(X)+F(-X)}D\right\|_{L_q}
    \le C\sqrt q\left(\sqrt{R/\rho}+1/R\right),
\]
using symmetry for $F(-X)$.  By \eqref{eq:bilinear-matrix-bound}, this
bounds the supremum over all feasible pairs with $a\le b$.
If $b<a$, transpose $\widehat X^W$ and reverse both coordinate orders;
this exchanges the two energy constraints in
\eqref{eq:endpoint-bilinear-supremum} and preserves the entry law.
The same argument applies with the endpoint sizes interchanged.
Taking the maximum of the two suprema, and then using
\eqref{eq:matrix-filling-comparison}, proves
\eqref{eq:bilinear-R-tradeoff}.

We now choose $R=\lceil\rho^{1/3}\rceil$.  For sufficiently large $B_1$,
the hypotheses $B_1\log^{3/2}(2n)\le\rho\le D$ ensure $R\ge4$,
$R^2\le D/64$, and
\[
    \rho/R\asymp\rho^{2/3}\ge Cq.
\]
Thus the conditions in \eqref{eq:subblock-parameter-conditions} hold, so
the moment bound \eqref{eq:bilinear-R-tradeoff} applies.  The terms
$\sqrt{R/\rho}$ and $1/R$ are now both $O(\rho^{-1/3})$,
giving $\|\mathcal A_W\|_{L_q}\le C\sqrt q\,\rho^{-1/3}$.
Substituting this moment estimate into
\eqref{eq:interval-maximum-from-moments} bounds
$\mathbb E\max_{W\in\mathcal W_D}\mathcal A_W$ and establishes
\cref{thm:local-quenched}.
\end{proof}

\section{From Transitive Tournaments to Arbitrary Tournaments}
\label{sec:arbitrary-tournaments}

We prove that every $n$-vertex tournament has a
$(1\pm\varepsilon)$-spectral sparsifier with
$\widetilde O(\min\{n^2,n/\varepsilon^3\})$ arcs.  Given a tournament $T$, we
order its vertices by nondecreasing outdegree and compare $T$ with the
transitive tournament $R$ determined by this order.  The energy $Q_T^+$ of $T$
is an exact sum of arc energies of $R$ with coefficients in $\{-1,+1\}$ and
the undirected energy of the pairs whose orientations disagree.  The coefficients record
whether the two orientations agree.  We extend the bound on
$\mathfrak G_{\TT_n}(\pi)$ from \cref{lem:tt-sampling-probabilities}
to the first term in
\cref{lem:signed-tt-sampling}.  We then combine this sampling with
effective-resistance sampling for the second term, using common indicators
so that the resulting sampled subgraph of $T$ has nonnegative weights.

\subsection{Energy domination after ordering by outdegree}

The purpose of this subsection is to prove that the directed energy of the
outdegree-ordered transitive tournament is at most an $O(\log n)$ factor
larger than that of the original tournament.  We do so by routing each
transitive-tournament arc along one- or two-arc paths of the original
tournament and bounding their congestion.

Let $T=(V,A)$ be an $n$-vertex tournament, and let $d_T^+(v)$ denote the
outdegree of $v$.  Fix an ordering
\begin{equation}
  d_T^+(v_1)\le d_T^+(v_2)\le \cdots\le d_T^+(v_n).
  \label{eq:score-order}
\end{equation}
breaking ties arbitrarily, and let $R$ be the transitive tournament with arcs
\begin{equation}
  v_j\longrightarrow v_i \qquad (i<j).
  \label{eq:score-backbone}
\end{equation}
Thus $R$ points from larger outdegree to smaller outdegree.  For every positive
integer $r$, write
\[
  H_r:=\sum_{k=1}^r \frac1k.
\]

The following reciprocal-degree inequality bounds the congestion of the
two-arc paths used in the proof of \cref{thm:score-order-domination}.

\begin{lemma}[Reciprocal degrees in a tournament]
\label{lem:tournament-reciprocal-degrees}
If $S$ is a tournament on $r$ vertices, then
\[
  \sum_{z\in V(S)}\frac{1}{d_S^+(z)+1}\le 2H_r,
  \qquad
  \sum_{z\in V(S)}\frac{1}{d_S^-(z)+1}\le 2H_r.
\]
\end{lemma}

\begin{proof}
Order the outdegrees as $s_1\le \cdots\le s_r$.  The first $k$ vertices span
$\binom{k}{2}$ arcs, and hence
\[
  k s_k\ge \sum_{i=1}^k s_i\ge \binom{k}{2}.
\]
Therefore $s_k\ge (k-1)/2$, and
\[
  \sum_{z\in V(S)}\frac{1}{d_S^+(z)+1}
  \le \sum_{k=1}^r \frac{2}{k+1}
  \le 2H_r.
\]
The indegree estimate follows by reversing every arc.
\end{proof}

\begin{theorem}[Energy domination for the outdegree order]
\label{thm:score-order-domination}
For every $x\in\mathbb R^V$,
\begin{equation}
  Q_R^+(x)\le \Lambda_n Q_T^+(x),
  \qquad
  \Lambda_n:=2(1+4H_n)=O(\log(2n)).
  \label{eq:score-domination}
\end{equation}
\end{theorem}

\begin{proof}
We route every arc of $R$ through $T$.  Consider an arc $u\to v$ of $R$.
If $u\to v$ also belongs to $T$, route its unit demand directly.  Otherwise
$T$ contains $v\to u$.  For a vertex $z$, let $N_T^+(z)$ and $N_T^-(z)$
denote its sets of outneighbors and inneighbors in $T$, respectively.  Define
\[
  C(u,v):=N_T^+(u)\cap N_T^-(v),
  \qquad
  \overline C(u,v):=N_T^+(v)\cap N_T^-(u).
\]
Comparing the outdegrees of $u$ and $v$ gives the exact identity
\begin{equation}
  d_T^+(u)-d_T^+(v)
  =|C(u,v)|-|\overline C(u,v)|-1.
  \label{eq:score-difference}
\end{equation}
Since $u$ appears after $v$ in the order in \eqref{eq:score-order}, the
left-hand side is
nonnegative.  Thus
\[
  |C(u,v)|\ge |\overline C(u,v)|+1\ge 1.
\]
For such a pair on which $T$ and $R$ disagree, route $1/|C(u,v)|$ units on
every path
$u\to w\to v$ with $w\in C(u,v)$.

We bound the congestion of this simultaneous routing.  Fix an arc $e=a\to b$
of $T$.  Let $X$ be the set of vertices $v$ for which $a\to b\to v$ is used
as a routed path for the $R$-arc $a\to v$.  Thus $v\to a$ in $T$, $v$ precedes
$a$ in the order in \eqref{eq:score-order}, and $b\to v$.  If $v,w\in X$ and
$v\to w$, then
$v\to w\to a$, so $w\in\overline C(a,v)$.  Consequently,
\[
  |C(a,v)|\ge d_{T[X]}^+(v)+1.
\]
Here $T[X]$ denotes the subtournament of $T$ induced by $X$.
The load placed on $e$ as a first leg is therefore at most
\[
  \sum_{v\in X}\frac{1}{|C(a,v)|}
  \le \sum_{v\in X}\frac{1}{d_{T[X]}^+(v)+1}
  \le 2H_n
\]
by Lemma~\ref{lem:tournament-reciprocal-degrees}.

For the second-leg load, let $Y$ be the set of vertices $u$
for which $u\to a\to b$ is used to route the $R$-arc $u\to b$.  If $u,w\in Y$
and $w\to u$, then $b\to w\to u$, and hence
$w\in\overline C(u,b)$.  Thus
\[
  |C(u,b)|\ge d_{T[Y]}^-(u)+1,
\]
and the second-leg load is also at most $2H_n$.  Direct routing contributes at
most one additional unit.  Hence every arc of $T$ has total routing load at
most
\begin{equation}
  1+4H_n.
  \label{eq:routing-congestion}
\end{equation}

Finally, for every directed two-arc path $u\to w\to v$,
\[
  (x_u-x_v)_+^2
  \le 2(x_u-x_w)_+^2+2(x_w-x_v)_+^2.
\]
Average this inequality over the paths used for each $R$-arc and sum over all
$R$-arcs.  The congestion bound in
\eqref{eq:routing-congestion} yields
\[
  Q_R^+(x)\le 2(1+4H_n)Q_T^+(x),
\]
as claimed.
\end{proof}

\subsection{An energy identity using pairs whose orientations disagree}

This subsection decomposes the energy of the original tournament into a
fixed-coefficient sum of arc energies of the ordered transitive tournament
and an undirected correction.  We also bound that correction by the original
directed energy, using $Q_R^+(x)\le\Lambda_n Q_T^+(x)$ from
\eqref{eq:score-domination}.

For each unordered pair, index the unique arc of $T$ by the arc of $R$ on the
same pair.  For $i<j$, put
\[
  q^R_{ij}(x):=(x_{v_j}-x_{v_i})_+^2,
  \qquad
  k_{ij}(x):=(x_{v_i}-x_{v_j})^2.
\]
Thus $q^R_{ij}$ is the contribution of the $R$-arc $v_j\to v_i$; for
$e=(v_j,v_i)\in E(R)$, we also write
$q_e^R(x):=q_{ij}^R(x)$ and $k_e(x):=k_{ij}(x)$.  Let
\[
  E_{\mathrm{dis}}
  :=\bigl\{\{v_i,v_j\}:i<j\text{ and }v_i\to v_j\text{ in }T\bigr\}
\]
be the set of unordered pairs on which $T$ and $R$ disagree, and let
$\mathcal F:=(V,E_{\mathrm{dis}})$ be the corresponding undirected graph.
Define
\[
  \sigma_{ij}:=
  \begin{cases}
    +1,&v_j\to v_i\text{ in }T,\\
    -1,&v_i\to v_j\text{ in }T.
  \end{cases}
\]
Finally, let
\[
  K_{\mathcal F}(x)
  :=\sum_{\{v_i,v_j\}\in E_{\mathrm{dis}}}k_{ij}(x)
\]
be the ordinary undirected quadratic energy of $\mathcal F$.

\begin{lemma}[An identity for $Q_T^+$ in terms of $q^R_{ij}$ and
$K_{\mathcal F}$]
\label{lem:tournament-feedback-decomposition}
For every $x\in\mathbb R^V$,
\begin{equation}
  Q_T^+(x)
  =\sum_{i<j}\sigma_{ij}q^R_{ij}(x)+K_{\mathcal F}(x).
  \label{eq:feedback-decomposition}
\end{equation}
Moreover,
\begin{equation}
  K_{\mathcal F}(x)
  \le Q_T^+(x)+Q_R^+(x)
  \le (\Lambda_n+1)Q_T^+(x).
  \label{eq:feedback-domination}
\end{equation}
\end{lemma}

\begin{proof}
If $T$ and $R$ agree on $\{v_i,v_j\}$, then the contribution of this pair to
$Q_T^+$ is $q^R_{ij}$.  If they disagree, then
\begin{equation}
  (x_{v_i}-x_{v_j})_+^2
  =(x_{v_i}-x_{v_j})^2-(x_{v_j}-x_{v_i})_+^2
  =k_{ij}-q^R_{ij}.
  \label{eq:disagreement-pair-energy}
\end{equation}
Summing these pairwise identities proves
\eqref{eq:feedback-decomposition}.  On every edge of $\mathcal F$,
$k_{ij}=(x_{v_i}-x_{v_j})_+^2+q^R_{ij}$.  Hence
\[
  K_{\mathcal F}(x)
  \le Q_T^+(x)+Q_R^+(x).
\]
The domination $Q_R^+\le\Lambda_n Q_T^+$ in
\eqref{eq:score-domination} gives the final inequality in
\eqref{eq:feedback-domination}.
\end{proof}

\subsection{Sampling the sum with fixed signs}

We next approximate the term $\sum_{i<j}\sigma_{ij}q^R_{ij}(x)$ in
\eqref{eq:feedback-decomposition}.  Its signs are fixed by the orientations
of $T$ and $R$ before sampling.  We use the probabilities constructed in
\cref{lem:tt-sampling-probabilities} and show that its bound on the expected
Gaussian width also holds after multiplying each arc contribution by a fixed
sign.  Let $0<\eta\le1/2$ be
the error parameter for this approximation.  The following lemma bounds the
expected supremum of the absolute sampling error normalized by $Q_R^+(x)$,
over all $x$ with $Q_R^+(x)>0$, by $O(\eta)$.  We continue to write
$q_e^R(x)=(x_{v_j}-x_{v_i})_+^2$ for the contribution of an arc
$e=(v_j,v_i)$ of $R$.

\begin{lemma}[Sampling with coefficients in $\{-1,+1\}$]
\label{lem:signed-tt-sampling}
Let $R$ be the $n$-vertex transitive tournament with arcs
$v_j\to v_i$ for $1\le i<j\le n$, and call $j-i$ the endpoint distance of
this arc.  There is a universal constant $C>0$ such
that, for every function $\sigma:E(R)\to\{-1,+1\}$ and every
$0<\eta\le 1/2$, there are inclusion probabilities
$a=(a_e)_{e\in E(R)}\in(0,1]^{E(R)}$ such that arcs with the same endpoint
distance receive the same probability and, writing $\sigma_e:=\sigma(e)$,
\begin{equation}
  \sum_{e\in E(R)}a_e
  \le \widetilde O\!\left(
    \min\left\{n^2,\frac{n}{\eta^3}\right\}
  \right)
  \label{eq:signed-tt-budget}
\end{equation}
and, for independent $\xi_e\sim\operatorname{Bernoulli}(a_e)$,
\begin{equation}
  \mathbb E\sup_{Q_R^+(x)>0}
  \frac{
    \left|
      \sum_{e\in E(R)}
      \sigma_e\left(\frac{\xi_e}{a_e}-1\right)q_e^R(x)
    \right|
  }{Q_R^+(x)}
  \le C\eta.
  \label{eq:signed-tt-error}
\end{equation}
\end{lemma}

\begin{proof}
Relabeling $v_j$ as $n+1-j$ identifies $R$ with $\TT_n$ and preserves
endpoint distances and directed energies.  Take the probabilities $\pi$
from \cref{lem:tt-sampling-probabilities} with target width $\eta$, and let
$a_e$ be the corresponding probability on each arc of $R$.  Then
\eqref{eq:tt-sampling-budget} gives \eqref{eq:signed-tt-budget}, and
\[
  \mathfrak G_R(a)=\mathfrak G_{\TT_n}(\pi)\le\eta.
\]
Let $(\xi'_e)_{e\in E(R)}$ be an independent copy of
$(\xi_e)_{e\in E(R)}$, put $\zeta_e:=\mathbf 1\{\xi_e\ne\xi'_e\}$, and
let $(g_e)_{e\in E(R)}$ be independent standard Gaussian variables,
independent of the two Bernoulli families $(\xi_e)_{e\in E(R)}$ and
$(\xi'_e)_{e\in E(R)}$.
Apply \cref{lem:sampling-gaussian-comparison} to the vectors
$(\sigma_e q_e^R(x)/Q_R^+(x))_{e\in E(R)}$ with $Q_R^+(x)>0$.
The left-hand side of \eqref{eq:signed-tt-error} is at most
\[
  C\mathbb E_{\zeta,g}\sup_{Q_R^+(x)>0}
  \left|
    \sum_{e\in E(R)}g_e\frac{\zeta_e}{a_e}
    \sigma_e\frac{q_e^R(x)}{Q_R^+(x)}
  \right|.
\]
Multiplying the independent standard Gaussian variables $g_e$ by the fixed
signs $\sigma_e$ preserves their joint distribution and independence from
$\zeta$.  Thus this expression equals $C\mathfrak G_R(a)\le C\eta$,
which proves \eqref{eq:signed-tt-error}.
\end{proof}

\subsection{Probability monotonicity and undirected sampling}

This subsection provides two tools for using common inclusion probabilities
for the signed sum $\sum_{i<j}\sigma_{ij}q^R_{ij}(x)$ and the undirected
energy $K_{\mathcal F}(x)$ in \eqref{eq:feedback-decomposition}.
First, increasing independent inclusion probabilities does not increase the
expectation of a convex function of the centered coefficients
$(\xi_e/p_e-1)_e$.  This lets us raise the probabilities from
\cref{lem:signed-tt-sampling} to meet the requirements of the undirected
sampling as well.  Second, we give a sampling guarantee that approximates
$K_{\mathcal F}$ uniformly over all potentials.

\begin{lemma}[Monotonicity under probability increase]
\label{lem:probability-increase}
Let $0<a_e\le p_e\le1$ for every $e$ in a finite set $E$, and let
$\Phi:\mathbb R^E\to[0,\infty]$ be convex.  If
$\xi_e^{(a)}\sim\operatorname{Bernoulli}(a_e)$ and
$\xi_e^{(p)}\sim\operatorname{Bernoulli}(p_e)$ each form an independent family,
then
\begin{equation}
  \mathbb E\,\Phi\!\left(
    \left(\frac{\xi_e^{(p)}}{p_e}-1\right)_{e\in E}
  \right)
  \le
  \mathbb E\,\Phi\!\left(
    \left(\frac{\xi_e^{(a)}}{a_e}-1\right)_{e\in E}
  \right).
  \label{eq:probability-monotonicity}
\end{equation}
\end{lemma}

\begin{proof}
First sample $\xi_e^{(p)}$.  Conditional on $\xi_e^{(p)}=1$, retain the
coordinate once more with probability $a_e/p_e$, using independent thinning
variables for the different coordinates; if $\xi_e^{(p)}=0$, discard it.
Denote the resulting indicator by $\xi_e^{(a)}$.  Then
\[
  \mathbb E\!\left[
    \left.\frac{\xi_e^{(a)}}{a_e}-1\,\right|\,\xi^{(p)}
  \right]
  =\frac{\xi_e^{(p)}}{p_e}-1.
\]
Apply conditional Jensen's inequality and then average over $\xi^{(p)}$.
\end{proof}

To approximate the undirected energy, we use independent sampling with
probabilities determined by effective resistances
\cite{SpielmanSrivastava2011}.  For an unweighted undirected graph
$F=(V,E_F)$ with Laplacian matrix $L_F$, the effective resistance of an
edge $e=\{u,v\}$ is
\begin{equation}
  \ell_e:=(\mathbf 1_u-\mathbf 1_v)^\top
    L_F^\dagger(\mathbf 1_u-\mathbf 1_v),
  \label{eq:undirected-effective-resistance}
\end{equation}
where $L_F^\dagger$ is the Moore--Penrose pseudoinverse of $L_F$.
Equivalently, with unit resistance on each edge of $F$, $\ell_e$ is the
voltage difference between $u$ and $v$ when one unit of current enters at
$u$ and leaves at $v$.  The next lemma uses these quantities to choose
edge inclusion probabilities.

\begin{lemma}[Independent spectral sampling]
\label{lem:independent-spectral-sampling}
There is a universal constant $C>0$ such that the following holds.  Let
$F=(V,E_F)$ be an unweighted undirected graph, put $n:=|V|$, and let
$0<\eta\le1/2$.
There are probabilities $b_e\in(0,1]$, $e\in E_F$, with
\begin{equation}
  \sum_{e\in E_F}b_e
  \le C\frac{n\log(2n)}{\eta^2},
  \label{eq:spectral-budget}
\end{equation}
such that the following holds.  Suppose that $b_e\le p_e\le1$.  Retain each
edge $e$ independently with probability $p_e$, and let $\xi_e$ denote its
indicator.  Then, with probability at least $7/8$,
\begin{equation}
  \left|
    \sum_{e=\{u,v\}\in E_F}
    \left(\frac{\xi_e}{p_e}-1\right)(x_u-x_v)^2
  \right|
  \le \eta\sum_{e=\{u,v\}\in E_F}(x_u-x_v)^2
  \qquad\forall x\in\mathbb R^V.
  \label{eq:spectral-error}
\end{equation}
\end{lemma}

\begin{proof}
Using $\ell_e$ from \eqref{eq:undirected-effective-resistance}, take
$b_e=\min\{1,C\eta^{-2}\log(2n)\ell_e\}$.  Since $F$ is unweighted,
$\sum_{e\in E_F}\ell_e=\operatorname{rank}(L_F)\le n-1$, which gives
\eqref{eq:spectral-budget}.

For each $e=\{u,v\}$, choose either order of its endpoints and put
$v_e:=(L_F^\dagger)^{1/2}(\mathbf 1_u-\mathbf 1_v)$.  If $P$ is the
orthogonal projector onto the image of $L_F$, then
$\sum_e v_ev_e^\top=P$ and $\|v_e\|_2^2=\ell_e$.
Set $\tau:=\eta^2/(C\log(2n))$ and
\[
  Y_e:=\left(\frac{\xi_e}{p_e}-1\right)v_ev_e^\top.
\]
These matrices are independent and centered.  If $p_e=1$, then $Y_e=0$;
otherwise $p_e\ge b_e$ implies $\ell_e/p_e\le\tau$.  Consequently,
\[
  \|Y_e\|_{\mathrm{op}}\le\tau,
  \qquad
  \sum_e\mathbb E Y_e^2
  =\sum_e\frac{1-p_e}{p_e}\ell_e v_ev_e^\top
  \preceq\tau P,
\]
where $\preceq$ denotes the positive semidefinite order.  Applying the
matrix Bernstein inequality \cite[Theorem~1.4]{Tropp2012} to $(Y_e)_e$ and
$(-Y_e)_e$ gives
\[
  \Pr\!\left[\left\|\sum_eY_e\right\|_{\mathrm{op}}>\eta\right]
  \le 2n\exp\!\left(-\frac{\eta^2}{2\tau+(2/3)\tau\eta}\right)
  \le\frac18
\]
for sufficiently large $C$.  On the complementary event, evaluating the
quadratic form of $\sum_eY_e$ on $L_F^{1/2}x$ gives
\eqref{eq:spectral-error} for every $x\in\mathbb R^V$.  This argument uses
only $p_e\ge b_e$, so it also covers coordinatewise increases in the
sampling probabilities.
\end{proof}

\subsection{Proof of the tournament theorem}

The identity \eqref{eq:feedback-decomposition} writes $Q_T^+(x)$ as
$\sum_{i<j}\sigma_{ij}q^R_{ij}(x)+K_{\mathcal F}(x)$.
We approximate the signed sum using \cref{lem:signed-tt-sampling}, with
error measured relative to $Q_R^+$, and the undirected term using
\cref{lem:independent-spectral-sampling}, with error measured relative to
$K_{\mathcal F}$.  The bounds \eqref{eq:score-domination} and
\eqref{eq:feedback-domination} compare both energies with $Q_T^+$, losing
only an $O(\log n)$ factor.  The construction below uses a common Bernoulli
indicator for the signed and undirected terms of each pair, producing a
$(1\pm\varepsilon)$-spectral sparsifier with
$\widetilde O(\min\{n^2,n/\varepsilon^3\})$ retained arcs.

\begin{proof}[Proof of \cref{thm:tournament-sparsification}]
Let $T=(V,A)$ be the tournament in the theorem, order its vertices as in
\eqref{eq:score-order}, and let $R$ be the transitive tournament in
\eqref{eq:score-backbone}.  Let
$\mathcal F=(V,E_{\mathrm{dis}})$ be the undirected graph whose edges are the
pairs on which $T$ and $R$ disagree.  For each
$e=(v_j,v_i)\in E(R)$, let $e_T$ be the unique arc of $T$ with endpoints
$v_i,v_j$, let $\bar e:=\{v_i,v_j\}$, and set $\sigma_e:=\sigma_{ij}$.
The map $e\mapsto\bar e$ is a bijection from $E(R)$ to the unordered vertex
pairs.  Put
\begin{equation}
  \eta:=\frac{\varepsilon}{C_{\mathrm{tour}}(\Lambda_n+1)},
  \label{eq:tournament-eta}
\end{equation}
where $C_{\mathrm{tour}}$ is a sufficiently large universal constant.

Apply Lemma~\ref{lem:signed-tt-sampling} to $R$ with signs $\sigma_e$ and error
parameter $\eta$, and obtain probabilities $a_e$.  Apply
Lemma~\ref{lem:independent-spectral-sampling} to the undirected graph
$\mathcal F$, obtaining probabilities $b_f$ for $f\in E_{\mathrm{dis}}$.
Set $b_f:=0$ for every other unordered pair $f$.  For every $e\in E(R)$, set
\begin{equation}
  p_e:=\max\{a_e,b_{\bar e}\}.
  \label{eq:common-probabilities}
\end{equation}
For each $e\in E(R)$, independently sample
$\xi_e\sim\operatorname{Bernoulli}(p_e)$.  If $\xi_e=1$, retain $e_T$ with
weight $1/p_e$.  Define the resulting random energy by
\[
  \widehat Q_T^+(x)
  :=\sum_{e\in E(R)}\frac{\xi_e}{p_e}q_{e_T}^+(x).
\]

Define
\[
  Z_R(x):=
  \sum_{e\in E(R)}
  \sigma_e\left(\frac{\xi_e}{p_e}-1\right)q_e^R(x),
\]
and
\[
  Z_{\mathcal F}(x):=
  \sum_{\substack{e\in E(R)\\ \bar e\in E_{\mathrm{dis}}}}
  \left(\frac{\xi_e}{p_e}-1\right)k_e(x).
\]
The functional
\[
  z\longmapsto
  \sup_{Q_R^+(x)>0}
  \frac{|\sum_{e\in E(R)}\sigma_e z_e q_e^R(x)|}{Q_R^+(x)}
\]
is convex.  Since $p_e\ge a_e$, Lemmas~\ref{lem:signed-tt-sampling} and
\ref{lem:probability-increase} imply
\[
  \mathbb E\sup_{Q_R^+(x)>0}
  \frac{|Z_R(x)|}{Q_R^+(x)}
  \le C\eta.
\]
Hence, by Markov's inequality, with probability at least $7/8$,
\begin{equation}
  |Z_R(x)|\le 8C\eta Q_R^+(x)
  \qquad\forall x.
  \label{eq:signed-good-event}
\end{equation}
Because $p_e\ge b_{\bar e}$ whenever $\bar e\in E_{\mathrm{dis}}$ and the
variables $\xi_e$ are independent,
Lemma~\ref{lem:independent-spectral-sampling} also gives, with probability at
least $7/8$,
\begin{equation}
  |Z_{\mathcal F}(x)|\le \eta K_{\mathcal F}(x)
  \qquad\forall x.
  \label{eq:spectral-good-event}
\end{equation}

For each $e\in E(R)$ whose orientation disagrees with $T$, multiply the
pairwise identity \eqref{eq:disagreement-pair-energy} by $\xi_e/p_e$ on
both sides.  On each agreeing pair, use $q_{e_T}^+(x)=q_e^R(x)$ with the
same coefficient $\xi_e/p_e$.
Consequently,
\begin{equation}
  \widehat Q_T^+(x)-Q_T^+(x)
  =Z_R(x)+Z_{\mathcal F}(x).
  \label{eq:sampled-feedback-decomposition}
\end{equation}
On the intersection of the events
\eqref{eq:signed-good-event} and
\eqref{eq:spectral-good-event}, the energy comparisons
\eqref{eq:score-domination} and \eqref{eq:feedback-domination} give
\begin{align*}
  |\widehat Q_T^+(x)-Q_T^+(x)|
  &\le 8C\eta Q_R^+(x)+\eta K_{\mathcal F}(x)\\
  &\le \eta\bigl(8C\Lambda_n+\Lambda_n+1\bigr)Q_T^+(x)\\
  &\le \varepsilon Q_T^+(x),
\end{align*}
where the last inequality uses the definition of $\eta$ in
\eqref{eq:tournament-eta} and a sufficiently large $C_{\mathrm{tour}}$.

It remains to count arcs.  The definition \eqref{eq:common-probabilities}
and the sampling budgets \eqref{eq:signed-tt-budget} and
\eqref{eq:spectral-budget} give
\begin{align*}
  \sum_{e\in E(R)}p_e
  &\le \sum_{e\in E(R)}a_e+\sum_{f\in E_{\mathrm{dis}}}b_f\\
  &\le \widetilde O\!\left(\frac{n}{\eta^3}\right)
       +O\!\left(\frac{n\log(2n)}{\eta^2}\right)\\
  &=\widetilde O\!\left(\frac{n}{\varepsilon^3}\right),
\end{align*}
because $\Lambda_n=O(\log(2n))$.  Markov's inequality shows that the number of
retained arcs is at most a universal constant times this expectation with
probability at least $7/8$.  The events in \eqref{eq:signed-good-event} and
\eqref{eq:spectral-good-event}, together with the event that the number of
retained arcs is at most a universal constant times its expectation, therefore
occur simultaneously with positive probability.  Fix such a realization.
Finally,
retaining every tournament arc gives the cap by $\binom n2$.
\end{proof}

\paragraph{Acknowledgements}

We thank Kam Chuen Tung for helpful discussions at an early stage of this work during his internship at the National Institute of Informatics.

\bibliographystyle{plain}
\bibliography{references}

\appendix

\section{Probabilistic Estimates for Bernoulli-Restricted Gaussian Matrices}
\label{app:prob-tools}

This appendix proves the moment estimates for the maximum retained-entry
count (\cref{lem:binomial-max}) and the operator norm of a
Bernoulli-restricted Gaussian matrix (\cref{lem:masked-gaussian-norm}).

\begin{proof}[Proof of \cref{lem:binomial-max}]
Stochastic domination and Bernstein's inequality give
$\Prob[Z_r>C(\mu+u)]\le e^{-u}$ for $u\ge0$, after changing $C$.
A union bound yields
\[
    \Prob\left[\max_r Z_r>C(\mu+\log(2M)+u)\right]\le e^{-u}.
\]
Integrating this tail using
$\mathbb E Y^t=t\int_0^\infty u^{t-1}\Prob[Y>u]\,du$
for $Y=\max_r Z_r$ proves the assertion.
\end{proof}

\begin{proof}[Proof of \cref{lem:masked-gaussian-norm}]
Let $\Delta$ be the maximum number of retained entries in any row or
column.  Since $m,k\le D$ and $p=\rho/D$, each count has mean at most
$2\rho$.  \Cref{lem:binomial-max} gives
\begin{equation}
    \|\Delta\|_{L_q}\le C(\rho+\log(2n)+q)\le C\rho.
    \label{eq:appendix-retained-count-moment}
\end{equation}
Condition on $\zeta$ and form the symmetric matrix
\[
    A=\begin{pmatrix}0&X\\X^\top&0\end{pmatrix}.
\]
We bound $\|X\|_{\rm op}^{2q}=\|A\|_{\rm op}^{2q}$ by
$\operatorname{tr}(A^{2q})$ and expand the trace as a sum over closed
walks of length $2q$ on the bipartite graph of retained entries.
Independence and the scalar identities
$\mathbb E g^{2s+1}=0$, $\mathbb E g^{2s}=(2s-1)!!$ show that the expected
product along a walk is $p^{-2q}$ times the number of pairings of its
$2q$ steps in which the two steps in each pair use the same edge.
There are $(2q-1)!!$ pairings.  For a fixed pairing, choose the starting
vertex in at most $m+k$ ways.  At the first occurrence of each pair,
choose an incident edge in at most $\Delta$ ways; at its second
occurrence, the edge is already fixed.  Thus there are at most
$(m+k)\Delta^q$ compatible walks for each pairing, and
\[
    \mathbb E_g\|X\|_{\rm op}^{2q}
    \le (m+k)(2q-1)!!\,\Delta^q/p^{2q}.
\]
Average over $\zeta$, take the $2q$th root, and use the bound on
$\|\Delta\|_{L_q}$ in \eqref{eq:appendix-retained-count-moment}, together with
$(2q-1)!!\le(2q)^q$ and $(m+k)^{1/(2q)}\le e^{1/2}$ to obtain
\[
    \bigl\|\|X\|_{\rm op}\bigr\|_{L_q}
    \le\bigl\|\|X\|_{\rm op}\bigr\|_{L_{2q}}
    \le C\sqrt{q\rho}/p=CD\sqrt{q/\rho}.
\]
\end{proof}

\end{document}